\documentclass{article}

\usepackage{microtype}
\usepackage{graphicx}
\usepackage{subcaption}
\usepackage{booktabs} 
\usepackage{tikz}
\usepackage{listings}
\usepackage{xcolor}

\usetikzlibrary{shapes.geometric, arrows.meta, positioning, calc, backgrounds, fit}

\usepackage{enumitem}
\usepackage{hyperref}

\usepackage[preprint]{icml2026}

\usepackage{amsmath}
\usepackage{amssymb}
\usepackage{mathtools}
\usepackage{amsthm}

\usepackage[capitalize,noabbrev]{cleveref}

\theoremstyle{plain}
\newtheorem{theorem}{Theorem}[section]
\newtheorem{proposition}[theorem]{Proposition}

\newtheorem{condition}[theorem]{Condition}
\theoremstyle{definition}
\newtheorem{definition}[theorem]{Definition}

\theoremstyle{remark}
\newtheorem{remark}[theorem]{Remark}

\usepackage[textsize=tiny]{todonotes}

\crefname{condition}{Condition}{Conditions}
\crefname{proposition}{Proposition}{Propositions}
\crefname{lemma}{Lemma}{Lemmas}
\crefname{theorem}{Theorem}{Theorems}

\icmltitlerunning{Submission and Formatting Instructions for ICML 2026}

\begin{document}

\twocolumn[
    \icmltitle{GRASP: Gradient Realignment via Active Shared Perception for Multi-Agent Collaborative Optimization}



  \icmlsetsymbol{equal}{*}

  \begin{icmlauthorlist}
    \icmlauthor{Sihan Zhou}{dlut}
    \icmlauthor{Tiantian He}{astar}
    \icmlauthor{Yifan Lu}{dlut}
    \icmlauthor{Yaqing Hou}{dlut}
    \icmlauthor{Yew-Soon Ong}{astar,ntu}
\end{icmlauthorlist}

\icmlaffiliation{dlut}{College of Computer Science and Technology, and Key Laboratory of Social Computing and Cognitive Intelligence (Ministry of Education), Dalian University of Technology, China}
\icmlaffiliation{astar}{Center for Frontier AI Research, Institute of High Performance Computing, Singapore Institute of Manufacturing Technology, Agency for Science, Technology and Research (A*STAR), Singapore}
\icmlaffiliation{ntu}{School of Computer Science and Engineering, Nanyang Technological University, Singapore}

\icmlcorrespondingauthor{Yaqing Hou}{houyq@dlut.edu.cn}

  \icmlkeywords{Multi-Agent Reinforcement Learning, Convergence Analysis, Non-Stationarity, Fixed-Point Theory, Game Theory, Policy Gradient Methods.}

  \vskip 0.3in
]



\printAffiliationsAndNotice{}  

\begin{abstract}
Non-stationarity arises from concurrent policy updates and leads to persistent environmental fluctuations. Existing approaches like Centralized Training with Decentralized Execution (CTDE) and sequential update schemes mitigate this issue. However, since the perception of the policies of other agents remains dependent on sampling environmental interaction data, the agent essentially operates in a passive perception state. This inevitably triggers equilibrium oscillations and significantly slows the convergence speed of the system. To address this issue, we propose Gradient Realignment via Active Shared Perception (GRASP), a novel framework that defines generalized Bellman equilibrium as a stable objective for policy evolution. The core mechanism of GRASP involves utilizing the independent gradients of agents to derive a defined consensus gradient, enabling agents to actively perceive policy updates and optimize team collaboration. Theoretically, we leverage the Kakutani Fixed-Point Theorem to prove that the consensus direction $u^*$ guarantees the existence and attainability of this equilibrium. Extensive experiments on StarCraft II Multi-Agent Challenge (SMAC) and Google Research Football (GRF) demonstrate the scalability and promising performance of the framework.
\end{abstract}
\section{Introduction}
Multi-Agent Reinforcement Learning (MARL)~\cite{icmlmarl1} has emerged as a powerful paradigm for addressing complex sequential decision-making problems involving multiple autonomous agents~\cite{zhou2025cooperative}. However, the transition from single-agent to multi-agent settings introduces a fundamental challenge, non-stationarity~\cite{nsqiu}. During the learning and interaction process, each agent perceives other agents as part of its environment. As all agents are simultaneously updating their policies, the environment from the perspective of any single agent becomes non-stationary~\cite{ns1}. This leads to a persistent policy interplay where policies are in constant flux, severely complicating the convergence of learning algorithms.

To mitigate this issue, the Centralized Training with Decentralized Execution (CTDE)~\cite{ctde} framework was proposed. By leveraging a centralized critic with access to global information, CTDE can provide more stable value estimates~\cite{qmix}. Despite employing a centralized critic to stabilize the training environment of the critic and guide the updates of the actor, the actor remains blind during policy updates. Decentralized actors simultaneously update their policies without explicit coordination, causing the equilibrium point and ultimate optimization goal of the entire multi-agent system to constantly shift~\cite{dynamicpoint}. This indicates that the non-stationarity issue in policy learning dynamics persists.

Another researches address non-stationarity through serialization mechanisms, categorizable into sequential updates and sequential execution. Methods like HAPPO~\cite{happo} employ sequential updates, discretizing the training process into time slices where only a single agent updates its policy at a time. While this creates a temporarily stationary environment locally~\cite{hasac,order}, the resulting stationarity is fragmented and discontinuous. During update intervals, the policy distributions of the remaining agents continue to evolve, leading to overall policy updates that remain synchronized across multiple overlapping time slices. This prevents smooth convergence of the joint policy~\cite{order2}. On the other hand, sequential decision methods~\cite{ace} construct weakly stationary environments by propagating the actions of preceding agents within a single time step. However, the actions passed represent only a stochastic point estimation of the preceding policy, making it difficult for subsequent agents to capture the full distribution of upstream policies and introducing high variance.

While these methods introduce a degree of coordination, they rely on what we term passive perception of policy updates~\cite{pass}. An agent is unaware of how other agents update their policies. It can only infer changes in policies of other agents by collecting new interaction data from the environment. This inference process is slow and inefficient. Particularly in the early stages of training, when policies change drastically, an agent may still be adapting to a previous policy of another agent that has already changed again~\cite{reddi2025recursive}. This means that there is no clear objective for optimizing agent policies in a Multi-Agent System (MAS), which can only guarantee that an equilibrium point always exists, but this point is not explicitly defined. 

To address these challenges and mitigate the impact of non-stationarity in multi-agent collaboration, we propose the Gradient Realignment via Active Shared Perception (GRASP) framework. This framework explicitly defines a generalized Bellman equilibrium as the stable objective for policy evolution. Specifically, the core of the GRASP framework lies in utilizing the independent gradients of each agent to compute a defined consensus gradient, enabling agents to actively perceive the policy evolution of other members and thereby optimize team collaboration. To provide a theoretical basis for the computation of consensus gradients, we employ Monotonic Improvement Analysis and Kakutani Fixed-Point Theorem to establish the necessary conditions for the consensus gradient operator $f(\cdot)$. In practical implementation, we introduce a quadratic programming (QP)~\cite{qp} model to solve for consensus gradients. Moreover, it demonstrates that the incentive-compatible consensus direction $u^*$ constructed from this QP formulation satisfies the derived key properties. This consensus direction leads to the existence of fixed points in the joint parameter space, thereby establishing the reachability of the generalized Bellman equilibrium. The main contributions of this paper are summarized as follows:
\begin{itemize}[nosep]
\item We propose the \textbf{GRASP framework}, which establishes incentive-compatible consensus by actively aligning individual updates with global gradients, ensuring collective consistency during synchronized learning.

\item We define \textbf{Generalized Bellman Equilibrium} to characterize Pareto-optimal steady states, formally proving its existence and convergence in non-stationary environments via Kakutani Fixed-Point Theorem.

\item We conduct extensive empirical evaluations across a diverse set of challenging benchmarks, including the StarCraft Multi-Agent Challenge (\textbf{SMAC}), its stochastic variant \textbf{SMACv2}, Google Research Football (\textbf{GRF}), and the Multi-Agent Particle Environment (\textbf{MPE}). Experimental results demonstrate that GRASP consistently outperforms state-of-the-art methods in both sample efficiency and asymptotic win rates.

\end{itemize}

\section{Related Work}

\subsection{Non-stationarity and Perception Modes in MARL}

Non-stationarity remains a fundamental obstacle in MARL. Existing research primarily attempts to mitigate this issue through two dimensions, spatial decomposition and temporal sequencing.
First, spatial dimension approaches, the CTDE paradigm decouples the information structure between learning and deployment.  Algorithms in this category, such as (Multi-Agent Deep Deterministic Policy Gradient) MADDPG~\cite{ctde}, employ a centralized critic that augments the local observation with global state information to stabilize value estimation.  Similarly, value decomposition methods like Value-Decomposition Networks (VDN) and Q-decomposition Multi-agent Independent eXtension (QMIX)~\cite{qmix} utilize a centralized mixing network (Mixer) to aggregate local utilities into a global objective $Q_{tot}$. These methods aim to construct a global state value assessment via centralized Critic or Mixer components, thereby creating a semi-stationary environment. However, this is fundamentally a compromise. While the Critic gains global information, the decentralized Actors still blindly update within a constantly drifting joint policy space. This weak stationarity does not resolve the dynamic interference at the policy learning level.

Second, temporal dimension approaches seek stationarity through serialization. For instance, HAPPO~\cite{happo} converts concurrent training into sequential training via serialized updates, while ACE~\cite{ace} introduces serialized decision-making at execution. Though these methods establish temporary stationarity within individual time slices or decision steps, their discrete, segmented stationarity fails to integrate into continuous, smooth optimization trajectories over global long-term cycles. In other words, local temporal stationarity cannot bridge global policy non-stationarity.


\subsection{Gradient Conflict and Consensus Optimization}

Reconciling divergent objectives is a central challenge in both Multi-Task Learning (MTL)~\cite{pcgradient} and MARL.  In the domain of MTL, prominent algorithms like  Projected Conflicting Gradients (PCGrad)~\cite{pcgradient} and Multiple Gradient Descent Algorithm (MGDA)~\cite{qp} have been developed to mitigate gradient interference.  These methods function by projecting conflicting gradients onto a normal plane or finding a minimum-norm vector within the convex hull, aiming to balance multiple tasks within a single model architecture. These methods provide different options for computing formula gradients within the GRASP framework we propose.

\subsection{Trust-Region Considerations and Fixed-Point Theory}

Ensuring stable convergence in high-dimensional policy spaces often necessitates trust-region methods, such as Trust Region Policy Optimization~\cite{trpo,hatrpo} and Proximal Policy Optimization (PPO)~\cite{ppo}, to prevent large, destructive updates that violate the validity of local linear approximations. These methods typically rely on heuristic $\epsilon$-clipping or Kullback-Leibler (KL) divergence constraints to restrict the magnitude of policy changes. While trust-region constraints guarantee monotonic improvement in single-agent settings, extending them to multi-agent environments proves challenging. The non-stationarity induced by concurrent updates causes the joint policy to shift, rendering the local approximation of the surrogate objective inaccurate and destabilizing the safe update region for individual agents~\cite{trpo}.

\section{Method}
In this section, we detail the theoretical foundations, algorithmic implementation, and convergence analysis of the proposed GRASP framework. ~\cref{method:formalization} introduces the extended Decentralized Partially Observable Markov Decision Process (Dec-POMDP)~\cite{pomdp}, defining the consensus gradient, its properties, and the Generalized Bellman Equilibrium. Subsequently, Section~\ref{method:grasp_mappo} describes the policy update methodology within the GRASP framework. The section concludes with a theoretical analysis in Sections~\ref{method:joint_policy_proof} and~\ref{method:value_function_proof}, which establish the convergence guarantees for both policy and value function updates.

\subsection{Formalization of Multi-Agent Policy Collaborative Update}
\label{method:formalization}
We model the multi-agent task as a stochastic game for $N$ agents, defined by the Dec-POMDP tuple $\mathcal{M} = {<N, \mathcal{S}, \boldsymbol{\mathcal{A}}, \boldsymbol{\mathcal{O}}, P, \mathcal{R}, \Omega,  \gamma , u^*>}$, where $N$ is the number of agents. 
$\mathcal{S}$ denotes the space of environmental states.
$\boldsymbol{\mathcal{A}}=\{{\mathcal{A}}_i\}_{i=1,2,...,N}$ represents the action spaces for all agents, where $\mathcal{A}_i$ is the action set of agent $i$. 
$\boldsymbol{\mathcal{O}}=\{\mathcal{O}_i\}_{i=1,2,...,N}$ 
  represents the observation spaces for all agents, where $\mathcal{O}_i$ is the observation set of agent $i$.
$P(s'|s,\mathbf{a}):\mathcal{S} \times \boldsymbol{\mathcal{A}}\times \mathcal{S} \rightarrow [0,1]$ represents the state transition function, defining the probability of transitioning to state $s'$ given the current state $s$ and the joint action  $\mathbf{a}=(a_1,a_2,...,a_N)$. The team reward function is $\mathcal{R}:\mathcal{S}\times \boldsymbol{\mathcal{A}}\rightarrow\mathbb{R}$, which maps a state and joint action to a shared reward. The observation function 
$\Omega(s, i):\mathcal{S}\rightarrow O_i$
determines the private observation $o^t_i \in O_i$
for agent $i$ at time 
$t$, which is derived as $o^t_i=\Omega(s^t, i)$. Finally, $\gamma \in[0,1]$ denotes the discount factor. Each agent aims to learn the joint policy $\{\pi_i(a_i|o_i)\}^N_{i=1}$ that maximizes the expected discounted global return $J(\theta)=\mathbb{E}[\sum^\infty_{t=0}\gamma^tR_{t}\sim\mathcal{R}(s^t,\{a_i^t\sim\pi_i(a_i|o_i,\theta_i)\}_{i=1}^N)]$, where $R_{t}$ is the team reward in timestep $t$. In the team reward setting, the optimization objective for individual agent $i$ is expressed as $J_i(\theta_i) = J(\theta_i, \boldsymbol{\theta}_{-i}) = \mathbb{E}_{\boldsymbol{\tau} \sim (\pi_i, \boldsymbol{\pi}_{-i})} \left[ \sum_{t=0}^{\infty} \gamma^t R_t \right].$ The $u^*$ represents the collective improvement direction for the team. In standard MARL, agent $i$ updates its policy using the local gradient $g_i(\theta) = \nabla_{\theta_i} J(\theta)$. However, the local gradient $g_i$ may conflict with the collective improvement direction $u^*$ of the team. This issue is because although the objective function $J$ is identical across agents, $g_i$ represents a partial derivative limited to local parameters. As a result, it fails to capture the simultaneous variations occurring in the rest of the team during the update phase. To address this issue, we seek a consensus direction $u^*$ that aligns individual updates.

 \begin{definition}[Optimal Collaboration Policy]
 \label{def1:opt_pi}
      $\pi^*_i\in{\{\pi_i\}^*}_{i=1}^{N}$, where ${\{\pi_i\}^*}_{i=1}^{N}$ denotes the optimal joint policy of the team, and $\pi^*_i$ is not an individual's optimal policy but rather the policy that maximizes the team's discount return $J(\theta)$.
 \end{definition}

\begin{theorem}[Gradient Realignment via Active Shared Perception Framework]
\label{the1:u_for_pi}
    When a $u^*$ satisfying the definition can be obtained, the MAS can solve for $\pi^*_i$ such that $\lim_{t \to \infty} \pi_i^{(t)} = \pi_i^*$, where $\pi_i^{(t)} = \pi_i^{(t-1)} + u^{(t)} + g_i^{(t)}, \quad t = 1, 2, \dots$. The $f(
    \cdot): \{g_1, \dots, g_N\} \to u^*$, we named the consensus operator that exists within this gradient update framework.
\end{theorem}

\begin{proposition}[Consensus Equilibrium as Generalized Bellman Equilibrium]
\label{the2:bellman}
    
The consensus equilibrium point reached by Theorem~\ref{the1:u_for_pi} ($u^* = \boldsymbol{0}$) can be interpreted as a \textbf{generalized Bellman equilibrium} in cooperative MARL. At this point, no joint policy improvement direction exists that simultaneously increases the expected returns of all agents.
\end{proposition}

From the formalization of the above problem and~\cref{the1:u_for_pi}, central to obtaining the optimal joint policy ${\{\pi_i\}^*}_{i=1}^{N}$ now lies in finding the consensus operator $f: \{g_1, \dots, g_N\} \to u^*$. Through Monotonic Improvement Analysis and an examination of the stability conditions required by Kakutani Fixed-Point Theorem, we conclude that the $u^*$ output by function $f$ must possess the following core properties:

\begin{condition}[Consensus Gradient Solver Operator]
\label{lem:Properties}
    \textbf{Pareto Improvement Property:} $u^*$ must ensure a non-negative projection on local gradients $g_i$ of all agents ($g_i^\top u^* \ge 0$).
    \textbf{Continuity:} The function $f$ must be continuous with respect to the input gradients $\{g_i\}$.
    \textbf{Consensus Representation:} $u^*$ should lie within the convex hull of all local gradients.
\end{condition}
In Condition~\ref{lem:Properties}, the Pareto Improvement Property guarantees that updating along $u^*$ does not harm the interests of any individual, forming the foundation for the monotonic convergence of the system. {Continuity} is a necessary mathematical prerequisite for proving the existence of consensus equilibrium points using the Kakutani Fixed-Point Theorem. {Consensus Representation} ensures $u^*$ is a legitimate representation of the intent of the team rather than an arbitrarily generated external vector.

Crucially, in GRASP, the computation of $g_i$ depends on the current policy, necessitating that sampling data be generated by the policy currently being updated. Consequently, our approach falls naturally within the on-policy learning paradigm, leading us to build upon the classic and state-of-the-art baseline, MAPPO.  We formally define the GRASP framework incorporating these PPO-based update rules in Section~\ref{method:grasp_mappo}. The subsequent sections are dedicated to the formal definition of the GRASP.

Nevertheless, our method is not strictly limited to on-policy settings. In off-policy contexts, although the $u^*$ calculated from replay buffer data may theoretically deviate from the true optimal direction due to distribution shift, it serves as a sufficiently effective approximation to drive policy improvement. In~\cref{ap:off-policy_grasp}, we detail the adaptation of GRASP for off-policy settings, propose a mechanism to rectify the estimation bias, and provide a theoretical proof of its convergence. Furthermore, we empirically validate these conclusions through comparative experiments and specific results are presented in~\cref{ap:qmix_grasp}.

%

\subsection{GRASP}
\label{method:grasp_mappo}
\subsubsection{Problem Formulation}
The GRASP framework operates within the CTDE paradigm. Under this paradigm, we adopt standard definitions from MARL to formulate the fundamental functions.

We formulate our fundamental objective using standard MARL definitions. Let $\pi_i(a_i | o_i; \theta_i)$ denote the policy of agent $i$, parameterized by $\theta_i$, which outputs a probability distribution over action $a_i$ given the local observation $o_i$~\cite{sutton1988learning}, where $\theta_i$ represents the policy parameters.
Used to evaluate the long-term expected return of states $s$, the state value function is defined as:
\begin{equation} 
V(s_t; \phi) \approx V^{\pi}(s_t) = \mathbb{E}_{\pi} \left[ \sum_{k=0}^{\infty} \gamma^k r_{t+k} \mid s_0 = s_t \right]. 
\end{equation}

Here, $\gamma$ is the discount reward factor. Note that $r_t$ represents the shared reward of the team. To reduce variance and enhance training stability, we employ Generalized Advantage Estimation (GAE)~\cite{gae}. First, we calculate the Temporal Difference (TD)~\cite{mnih2015human} error $\delta_t$:
\begin{equation} 
\delta_t = r_t + \gamma V(s_{t+1}; \phi) - V(s_t; \phi).
\end{equation}

Subsequently, the advantage function is calculated recursively:
\begin{equation} 
\hat{A}_t = \sum_{l=0}^{\infty} (\gamma \lambda)^l \delta_{t+l}. 
\end{equation}
Here, $\lambda \in [0, 1]$ represents the smoothing parameter of GAE. Specifically, when $\lambda=0$, GAE degenerates into the TD error. $l$ denotes the time offset relative to the current time step $t$.

\subsubsection{Actor Update via GRASP}
The actor update serves as the core component of GRASP. It determines the consensus direction via an operator satisfying Condition~\ref{lem:Properties}. Specifically, it instantiates the representative operator defined in Proposition~\ref{the2:KKT} and approximates the aforementioned trust-region constraints utilizing the clipping mechanism of PPO. The process is divided into three steps.

Step 1: Local Gradient Computation. Each agent first computes its private expected policy gradient, representing the individual optimal update direction:
\begin{equation} 
\begin{aligned} 
g_i(&\theta_{i, \text{old}}) = \\
    &\mathbb{E}_{o_{i,t}, a_{i,t} \sim \pi_{i, \text{old}}} \left[ \nabla_{\theta_i} \log \pi_i(a_{i,t} | o_{i,t}; \theta_i) \cdot \hat{A}_t \right].
\end{aligned}
\end{equation}

Step 2: Consensus Gradient Determination.
To resolve gradient conflicts and identify the Pareto-optimal direction, GRASP solves for the minimum-norm vector $u^*$ within the Convex Hull of the local gradients of all agents:
\begin{equation} 
u^*(\theta_{\text{old}}) = \operatorname*{argmin}_{u \in \operatorname{ConvexHull}({g_1, \dots, g_N})} \|u\|^2.\end{equation}

Step 3: Gradient Realignment and Safe Update. We inject the consensus gradient $u^*$ into the standard PPO update process, performing gradient ascent along $v = g_i+u^*$, a hybrid gradient direction. First, we define the standard PPO clipped surrogate objective:
\begin{equation}
\begin{aligned}
    J_i^{\text{PPO}}(\theta_i) =
     &\mathbb{E}_{t} \big[ \min \big( \rho_t(\theta_i) \hat{A}_t, \\  &\text{clip}(\rho_t(\theta_i), 1-\epsilon, 1+\epsilon) \hat{A}_t \big) \big],
\end{aligned}
\end{equation}

where $\rho_t(\theta_i) = \frac{\pi_{\theta_i}(a_{i,t}|o_{i,t})}{\pi_{\theta_{i, \text{old}}}(a_{i,t}|o_{i,t})}$ is the probability ratio~\cite{ppo}. It is used to limit the step size of policy updates, preventing the policy from undergoing excessively large or drastic changes.

Then, we construct the final GRASP update rule. The update direction $\Delta \theta_i$ for parameter $\theta_i$ is jointly determined by the original PPO gradient and the weighted consensus gradient $u^*$:
\begin{equation}
    \theta_{i, \text{new}} \leftarrow \theta_{i, \text{old}} + \eta \cdot \underbrace{\left( \nabla_{\theta_i} J_i^{\text{PPO}}(\theta_i) \bigg|_{\theta_i = \theta_{i, \text{old}}} + u^* \right)}_{\text{Realigned Update Direction}},
\end{equation}

where $\eta$ is the learning rate. This can be equivalently expressed as maximizing the following hybrid objective function:
\begin{equation}
    \theta_{i, \text{new}} = \operatorname*{argmax}_{\theta_i} \left( J_i^{\text{PPO}}(\theta_i) +  \langle \theta_i, u^* \rangle \right).
\end{equation}

\subsubsection{Critic Update}
The objective of the critic is to approximate the value function of the current policy $\pi_{\text{old}}$. The target return $\hat{R}_t$ is constructed using the estimated advantage $\hat{A}_t$ and the value estimate $V(s_t; \phi)$ (i.e., $\hat{R}_t = \hat{A}_t + V(s_t; \phi_{\text{old}})$). The critic then minimizes the Mean Squared Error (MSE) between its prediction and this target return, as follows:
\begin{equation} 
\begin{aligned} \mathcal{L}(\phi) = \mathbb{E}_{t} \Bigg[ &\max \Bigg( \left( V(s_t; \phi) - \hat{R}_t \right)^2, \\ &\bigg( \text{clip}\Big(V(s_t; \phi_i), V(s_t; \phi_{ \text{old}}) - \varepsilon, \\ &\ \  \ \ \ \ \ \ \ \  V(s_t; \phi_{\text{old}}) + \varepsilon\Big) - \hat{R}_t \bigg)^2 \Bigg) \Bigg].
\end{aligned} 
\end{equation}
The update rule is:
\begin{equation}
    \phi_i \leftarrow \phi_i - \eta \nabla_{\phi_i} \mathcal{L}(\phi).
\end{equation}
\subsubsection{Constraint Consistency}
It is noteworthy that in practical engineering implementation, both the policy and state value functions undergo multi-epoch updates on a single sampling batch. Although the consensus gradient $u^*$ is introduced as a bias term within the objective function, the probability ratio $\rho_t$ is recomputed based on the latest $\theta_i$ during each subsequent PPO epoch update. This implies that even with external gradient guidance, the PPO clip mechanism remains active, ensuring that the final policy update is restricted within the trust region. A detailed theoretical compatibility analysis regarding the trust-region constraint and the consensus direction constraint is provided in  Subsection~\ref{grasp_ppo_trust}. The overall algorithm is described in~\cref{alg:grasp} in~\cref{ap:alg}.
\subsection{Joint Policy Functions Monotonic Improvement and Convergence under GRASP}
\label{method:joint_policy_proof}
\subsubsection{Existence of Consensus Equilibrium Points}
With the well established Kakutani Fixed-Point Theorem~\cite{kakutani}, we derive the following theorem to theoretically guarantee that the stable learning dynamics of GRASP possess a well-defined equilibrium point.
\begin{theorem}[GRASP Equilibrium Point Existence]
\label{th:eq_p_ex}
Suppose the joint policy parameter space $\Theta$ is a nonempty, compact, and convex subset of the Euclidean space $\mathbb{R}^D$. Then, the GRASP learning process possesses a fixed point $\theta^* \in \Theta$ that is a consensus equilibrium point.
\end{theorem}
\begin{proof}
The proof proceeds by verifying that the update mapping defined by the GRASP framework satisfies the necessary conditions of the Kakutani Fixed-Point Theorem—specifically, the properties of upper hemi-continuity and convexity on the compact space $\Theta$. Consequently, the existence of the consensus equilibrium point $\theta^*$ is guaranteed. The detailed derivation is provided in Appendix~\ref{ap:Kakutani Fixed-Point Theorem}.
\end{proof}
This resembles the state in single-agent RL where the Bellman optimality equation holds, the policy is already optimal, and no policy change would enhance value. Thus, GRASP is a robust iterative method for finding such stable fixed points.
\subsubsection{Trust-Region Constrained Consensus Optimization}
\label{grasp_ppo_trust}

Distinguished from standard MARL frameworks, GRASP models the multi-agent cooperative learning problem as a Trust-Region Constrained Consensus Optimization (TRCCO) problem. In the $k$-th iteration, given the old policy parameters $\theta_{\text{old}}$, GRASP aims to find a parameter update $\Delta \theta$ that maximizes alignment with the Pareto-optimal consensus direction $u^*$ while strictly remaining within the trust region to guarantee monotonic performance improvement. This is formalized as:
\begin{equation}
    \begin{aligned}
\max_{\Delta \theta} \quad & \langle u^*(\theta_{\text{old}}), \Delta \theta \rangle \\
\text{s.t.} \quad & \mathbb{E}_{s \sim \rho^{\boldsymbol{\pi}_{\text{old}}}} \big[ D_{KL}(\boldsymbol{\pi}_{\theta_{\text{old}}}(\cdot|s) \ || \\ & \ \ \ \ \ \ \ \ \ \ \ \ \ \ \ \ \ \ \ \ \ \   \boldsymbol{\pi}_{\theta_{\text{old}} + \Delta \theta}(\cdot|s)) \big] \le \delta.
\end{aligned}
\end{equation}
Here, $u^*(\theta_{\text{old}})$ represents the optimal cooperative gradient direction, and the $D_{KL}$ constraint defines a safe neighborhood within the parameter space. Solving the KL-divergence constraint explicitly involves high-order Hessian computations, which are computationally expensive. Following the PPO paradigm, we relax this constraint by incorporating the standard clipped surrogate objective and injecting the consensus direction as a linear regularization term. The unconstrained optimization problem in GRASP is thus reformulated as maximizing the following hybrid objective:
\begin{equation}
    \label{eq:grasp_update}
    \theta_{k+1} = \operatorname*{argmax}_{\theta} \left( J^{\text{PPO}}(\theta) +  \langle \theta, u^* \rangle \right),
\end{equation}
where $J^{\text{PPO}}(\theta) = \mathbb{E}_{t} [ \min ( \rho_t(\theta) \hat{A}_t, \operatorname{clip}(\rho_t(\theta), 1-\epsilon, 1+\epsilon) \hat{A}_t ) ]$ denotes the standard PPO objective. 

This formulation fundamentally alters the gradient field by updating the parameters along a realigned trajectory guided by $u^*$, rather than following the naive steepest ascent direction $\sum \nabla J_i$. To theoretically characterize the stability of this update, we propose the following proposition:

The update rule derived from \cref{eq:grasp_update} guarantees a dual-layer stability property termed Safe Pareto-Optimality. The consensus term $ u^*$ ensures that the update direction lies strictly within the Pareto descent cone of the joint loss surface, effectively filtering out conflicting gradient components that would otherwise lead to oscillation. Despite the introduction of the consensus bias, the PPO clipping mechanism imposes a strict bound on the policy probability ratio $\rho_t(\theta)$. This guarantees that the valid update remains within the trust region intersection defined by $\epsilon$, preserving the monotonic improvement property.


\subsection{Value Function Monotonic Improvement and  Convergence under GRASP}
\label{method:value_function_proof}

To formally establish the convergence of the value function under the proposed update rule, we analyze the contraction properties of the GRASP-Bellman operator.

\begin{proposition}[Consensus-Driven Update Mechanism]
\label{def:monotomic_policy_improvent}
Let $J(\theta) = \mathbb{E}_{\tau \sim \pi_{\theta}} \left[ \sum_{t=0}^{\infty} \gamma^t r_t \right]$ denote the joint objective function with local gradients $g_i = \nabla_{\theta_i} J(\theta)$. The GRASP update is defined as $\Delta \theta_i = \eta (g_i +  u^*)$, where $u^*$ is the minimum-norm element in the convex hull of $\{g_i\}_{i=1}^N$.
\end{proposition}

\begin{proof}
The proof proceeds by analyzing the first-order approximation of the joint objective function update. We first utilize the geometric properties of the minimum-norm consensus vector $u^*$ within the convex hull to guarantee a strictly non-negative improvement in the joint objective $J(\theta)$ via Taylor Expansion. Subsequently, by invoking the Policy Improvement Theorem and performing a standard recursive telescoping expansion, we establish that this local objective ascent translates into a uniform global improvement of the value function $V(s)$. The detailed derivation is provided in Appendix~\ref{ap:value_convergence}.
\end{proof}
The convergence of the value function within the GRASP framework is provided in~\cref{ap:convergence_value_function}.

\begin{figure*}[ht!]
    \centering
    \includegraphics[width=0.9\linewidth]{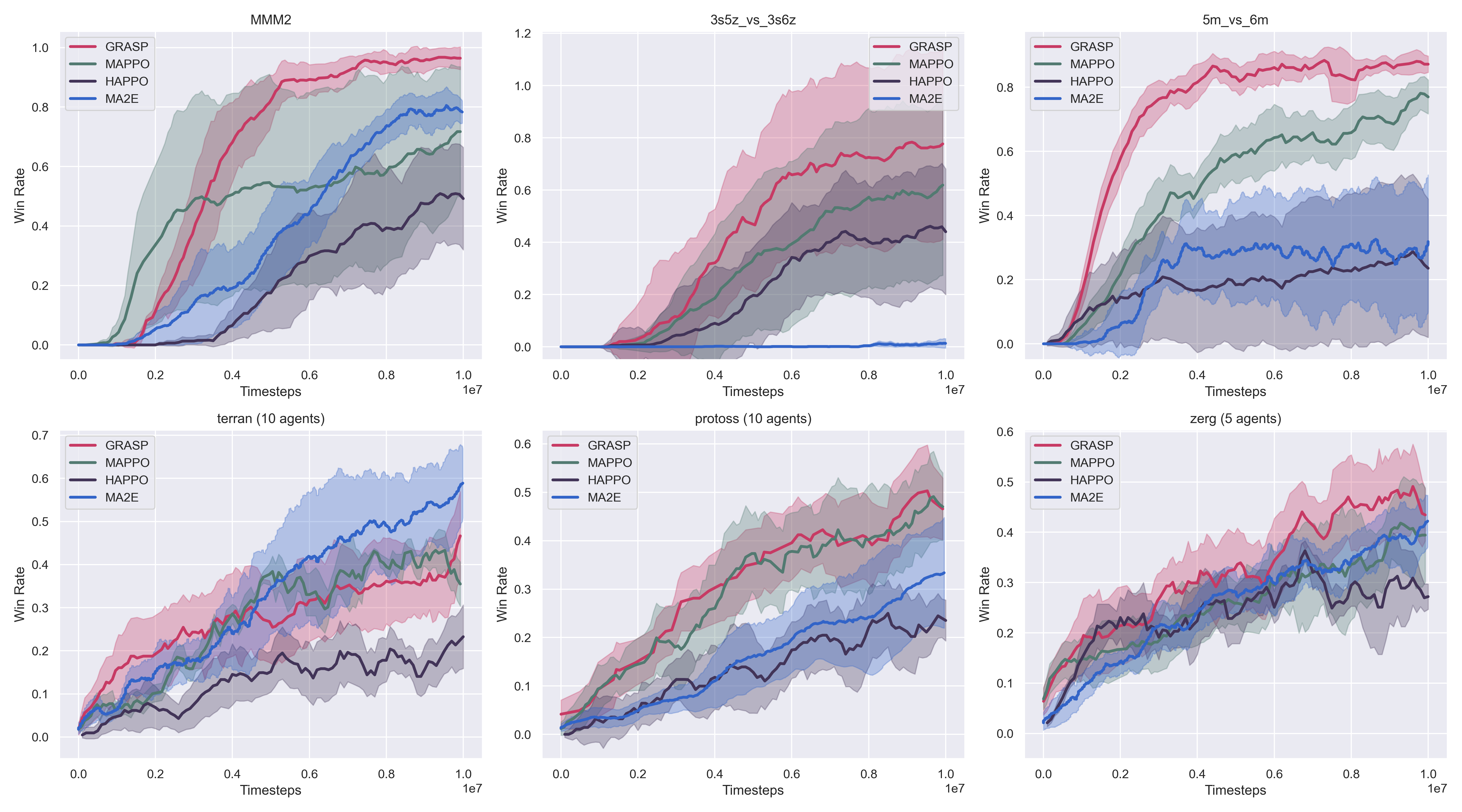}
    \caption{Comparison of win rate learning curves between GRASP and baseline algorithms (MAPPO, HAPPO, and $\text{MA}^2\text{E}$) across six scenarios in SMAC and SMACv2. Solid lines represent the average win rate from multiple runs, while shaded areas indicate standard deviation.}
    \label{fig:smac_res}
\end{figure*}

\begin{table*}[h!]
\centering
\caption{Win rate comparison (Mean $\pm$ Std) on different scenarios. The best results are highlighted in \textbf{bold}.}
\label{tab:win_rate_comparison}
\small
\setlength{\tabcolsep}{8pt}

\begin{tabular}{lcccc}
\toprule
\textbf{Scenario} & \textbf{GRASP (Ours)} & \textbf{MAPPO} & \textbf{HAPPO} & \textbf{$\text{MA}^2\text{E}$} \\ 
\midrule
MMM2 & \textbf{0.967 $\pm$ 0.031} & 0.717 $\pm$ 0.218 & 0.509 $\pm$ 0.168 & 0.839 $\pm$ 0.021 \\
3s5z vs 3s6z & \textbf{0.783 $\pm$ 0.360} & 0.619 $\pm$ 0.345 & 0.461 $\pm$ 0.232 & 0.027 $\pm$ 0.010 \\
5m vs 6m & \textbf{0.909 $\pm$ 0.044} & 0.815 $\pm$ 0.075 & 0.350 $\pm$ 0.321 & 0.372 $\pm$ 0.114 \\
Terran (10 agents) & 0.467 $\pm$ 0.103 & 0.571 $\pm$ 0.000 & 0.233 $\pm$ 0.074 & \textbf{0.583 $\pm$ 0.057} \\
Protoss (10 agents) & \textbf{0.563 $\pm$ 0.045} & 0.491 $\pm$ 0.094 & 0.250 $\pm$ 0.069 & 0.391 $\pm$ 0.054 \\
Zerg (5 agents) & \textbf{0.491 $\pm$ 0.084} & 0.418 $\pm$ 0.083 & 0.363 $\pm$ 0.020 & 0.396 $\pm$ 0.065 \\ 
\bottomrule
\end{tabular}%

\end{table*}


\subsection{Construction and Proof of the Consensus Gradient Function}
Based on the above analysis, we propose solving the QP problem of minimizing the convex combination of the gradient norm~\cite{qp} to find $u^*$ as the operator $f$. Intuitively, this is geometrically equivalent to finding the vector closest to the origin within the convex hull formed by the gradients.

\begin{proposition}[QP Consensus Operator]
\label{the2:KKT}
    Define the function $f(\cdot)$ as solving the following optimization problem:
    $$u^* = \sum_{i=1}^N c_i^* g_i, \quad \text{where } \boldsymbol{c}^* = \arg\min_{\boldsymbol{c} \in \Delta^N} \frac{1}{2} \left\| \sum_{i=1}^N c_i g_i \right\|_2^2$$
\begin{equation}
\label{qp_opt}
     and~\text{s.t. } \sum_{i=1}^N c_i = 1, \quad c_i \ge 0.
\end{equation}
By the definition of the optimization problem, $u^*$ satisfies consensus representativeness and necessarily satisfies Pareto improvement, i.e., for any agent $j$, we have $g_j^\top u^* \ge \|u^*\|_2^2 \ge 0$.
\end{proposition} 
\begin{proof}
    This optimization problem is a convex quadratic programming problem with linear constraints. We can analyze the properties of its optimal solution using the Karush-Kuhn-Tucker (KKT)~\cite{convex} conditions. The detailed proof process is provided in Appendix~\ref{apd:kkt_proof}.
\end{proof}
Through the above proof, we confirm that $u^*$, the solution to the QP problem, possesses a non-negative projection onto the local gradient $g_j$ for all agents. This demonstrates that $u^*$ constitutes a strictly defined Pareto improvement direction. Furthermore, by Berge's Maximum Theorem, the unique solution to this convex QP problem is continuous with respect to gradients, thereby satisfying the continuity condition required for fixed-point analysis. Furthermore, the QP solver is an operator that satisfies Condition~\ref{lem:Properties}. The QP algorithm is described in~\cref{alg:qp_solver} in~\cref{ap:alg}.


\section{Experiment}

We evaluated GRASP in the StarCraft Multi-Agent Challenge (SMAC)~\cite{smac}, SMACv2~\cite{smacv2}, and Google Research Football (GRF)~\cite{grf} environments. SMAC is one of the most popular MARL benchmarks, covering a wide range of cooperative micro-control scenarios. We conducted experiments on SMAC HARD and SuperHARD scenarios. SMACv2 complements the deterministic nature of SMAC by randomizing starting positions and unit types, while also altering units' line of sight and attack ranges. In GRF, multiple agents collaborate to play a soccer match. All experiments were run within $1 \times 10^7$ time steps per trial, and we report the average win rate and shadow standard error across three distinct random seeds. The experimental parameters for comparison are provided in \cref{ap:settingup}.

\subsection{GRASP Collaborative MARL}




First, we integrate GRASP into the MARL framework by selecting MAPPO~\cite{mappo}, the most representative SOTA method among on-policy gradient methods, to form the novel GRASP-MAPPO method. We compare GRASP against multiple MARL algorithms, including MAPPO, HAPPO~\cite{happo}, and $\text{MA}^2\text{E}$~\cite{ma2e}. MAPPO and HAPPO are on-policy methods, while $\text{MA}^2\text{E}$ is an off-policy method. We validated our approach on the most challenging scenarios of SMAC: 5m\_vs\_6m, MMM2, and 3s5z\_vs\_3s6z. Despite 5m\_vs\_6m being a notoriously difficult map, most existing algorithms still fail to achieve high win rates. We constructed three scenarios within SMACv2 to assess GRASP's performance, adopting SMACv2 to overcome limitations SMAC's limitations, such as a lack of randomness. Unlike SMAC, SMACv2 incorporates randomly generated unit starting positions, introducing greater randomness to create challenging scenarios.

Fig.~\ref{fig:smac_res} and Tab.~\ref{tab:win_rate_comparison} illustrate the performance comparison between the GRASP model and baselines, with the GRASP-MAPPO method representing our proposed approach. In the HARD and SuperHARD scenarios of SMAC, GRASP achieves higher sampling efficiency and upper bounds on win rates compared to all baselines, demonstrating a significant improvement over MAPPO. In SMACv2, although GRASP underperforms $\text{MA}^2\text{E}$ on the 10gen\_terran scenario, it outperforms $\text{MA}^2\text{E}$ across all other scenarios. Crucially, GRASP consistently outperforms MAPPO because it is a framework adaptable to all other CTDE-based MARL methods. The key strength of GRASP lies in its ability to enhance baseline methods.

\subsection{Extensibility of GRASP}

\begin{figure*}[ht!]
    \centering
    \includegraphics[width=0.9\linewidth]{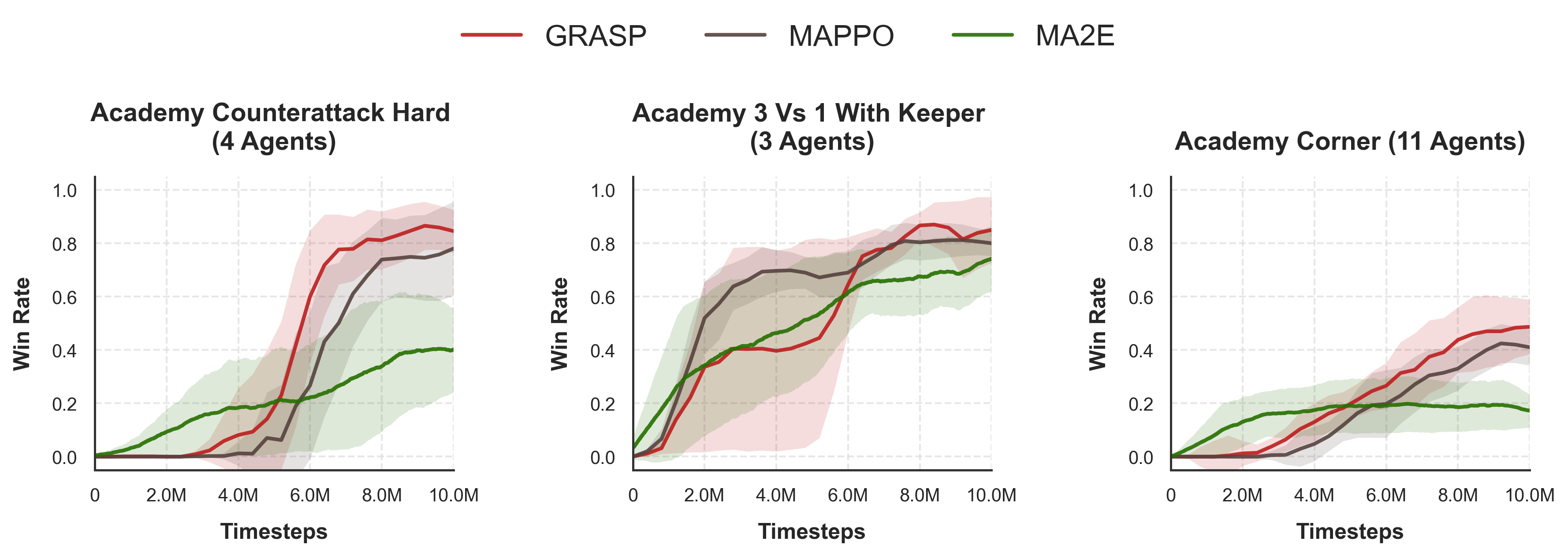}
    \caption{Comparison of win rate learning curves between GRASP and baseline algorithms (MAPPO and $\text{MA}^2\text{E}$) across three scenarios in GRF. Solid lines represent the average win rate from multiple runs, while shaded areas indicate standard deviation.}
    \label{fig:football _res}
\end{figure*}
\begin{figure}[ht!]
    \centering
    \includegraphics[width=0.65\linewidth]{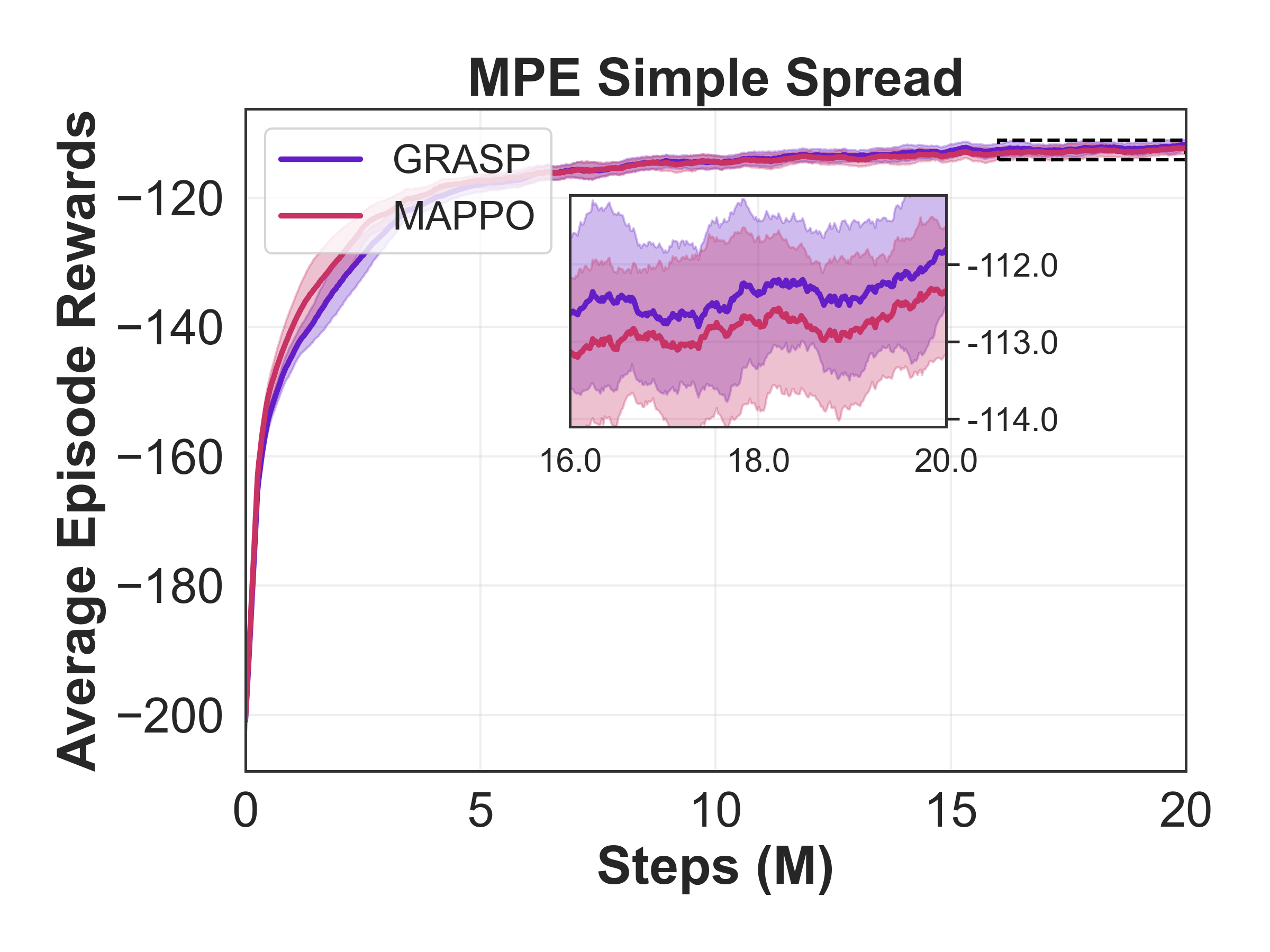}
    \caption{Comparison of win rate learning curves between GRASP and MAPPO in the Simple Spread scenario of MPE. Solid lines represent the average win rate from multiple runs, while shaded areas indicate standard deviation.}
    \label{fig:grasp_mpe_res}
\end{figure}

   
   
To comprehensively validate the applicability and robustness of the GRASP framework in complex dynamic environments, we conducted extensive experiments on the GRF platform. GRF is a physics-based, high-fidelity multi-agent simulation environment where the core winning condition is scoring a goal in the net of the opponent. Unlike traditional maze or simple particle environments, GRF features a high-dimensional state space, extremely sparse reward signals, and highly random opponent policies. This demands that agents not only possess precise individual control capabilities but also achieve high policy consistency to execute intricate passing and cutting coordination, making it the ideal testing ground for GRASP consensus gradients.

In selecting specific scenarios to cover collaborative challenges of varying difficulty and scale, we chose three representative official scenarios: Academy Counterattack Hard, Academy 3 Vs 1 With Keeper, and Academy Corner. The characteristics of these three scenarios are shown in Appendix~\ref{ap:settingup} (Tab.~\ref{tab:grf_scenarios}).

Fig.~\ref{fig:football _res} shows the win rate comparison curves of GRASP against the current SOTA method MAPPO and the exploration-focused $\text{MA}^2\text{E}$ method. Experimental results demonstrate that, thanks to the gradient realignment mechanism enabled by active perception, GRASP exhibits faster convergence and higher asymptotic win rates across all scenarios. When baseline methods stabilize and encounter suboptimal policy bottlenecks, GRASP effectively breaks through this upper bound, achieving significantly higher asymptotic win rates. This proves its capability to discover superior cooperative joint policy.

Furthermore, to validate the universality of GRASP across different collaborative paradigms, we conducted additional experiments in the Multi-Agent Particle Environments (MPE)~\cite{ctde}  setting. Unlike the emphasis of GRF on unified objectives and sparse rewards, MPE scenarios employ an additive reward mechanism, where the total team reward is the sum of individual rewards earned by each agent. This mechanism requires agents to strictly avoid physical collisions or mutual interference while pursuing individual optimization, thereby imposing stringent demands on the conflict resolution capabilities of joint policy.

Fig.~\ref{fig:grasp_mpe_res} presents the comparison results between GRASP and MAPPO in the Simple Spread scenario. Given the wide reward value range in the early training stages of this scenario, a local zoom-in view is provided to more clearly illustrate performance differences during the convergence phase. Results show that, at the same number of training steps, GRASP converges to a superior policy solution compared to MAPPO. This further validates the advantage of GRASP in handling multi-agent policy conflicts, effectively guiding the team toward more efficient, conflict-free collaboration.

We provide additional validations in Appendix~\ref{ap:MPE_2730}. Specifically, we challenge GRASP in the super-hard 27m\_vs\_30m scenario of SMAC to demonstrate its scalability and robustness in large-scale, high-dimensional multi-agent coordination tasks. Finally, Appendix~\ref{ap:qmix_grasp} provides a detailed discussion regarding the limitations of GRASP, particularly its reduced consensus efficacy when applied to off-policy methods (e.g., QMIX) due to stale transitions.

\section{Conclusion}
In this paper, we propose the GRASP framework, which incorporates a novel policy update method and a property definition for the consensus gradient operator. Theoretically, we demonstrate that the proposed policy update rule, driven by consensus gradients satisfying these properties, ensures the policy evolution process meets the conditions of Kakutani Fixed Point Theorem. This theoretically guarantees convergence toward the generalized Bellman Equilibrium. Extensive experiments on SMAC, SMACv2, GRF, and MPE consistently show that GRASP achieves state-of-the-art performance with superior scalability and generality. In highly stochastic environments, increased gradient variance expands the convex hull, which diminishes the magnitude of the consensus update. While this safely prevents drastic policy fluctuations, it attenuates the overall consensus gains. Future work will explore hybrid consensus operators combining representations and gradients for robust estimation, maintaining efficacy under extreme stochasticity.

\section{Impact Statement}

This paper presents work whose goal is to advance the field of Machine
Learning. There are many potential societal consequences of our work, none
which we feel must be specifically highlighted here.


\bibliography{example_paper}
\bibliographystyle{icml2026}

\newpage
\appendix
\onecolumn

\crefalias{section}{appendix}
\crefalias{subsection}{appendix}

\section{Additional Experiment}
\subsection{Transferability of GRASP}
\label{ap:qmix_grasp}

\begin{figure*}[ht]
    \centering
    \includegraphics[width=1\linewidth]{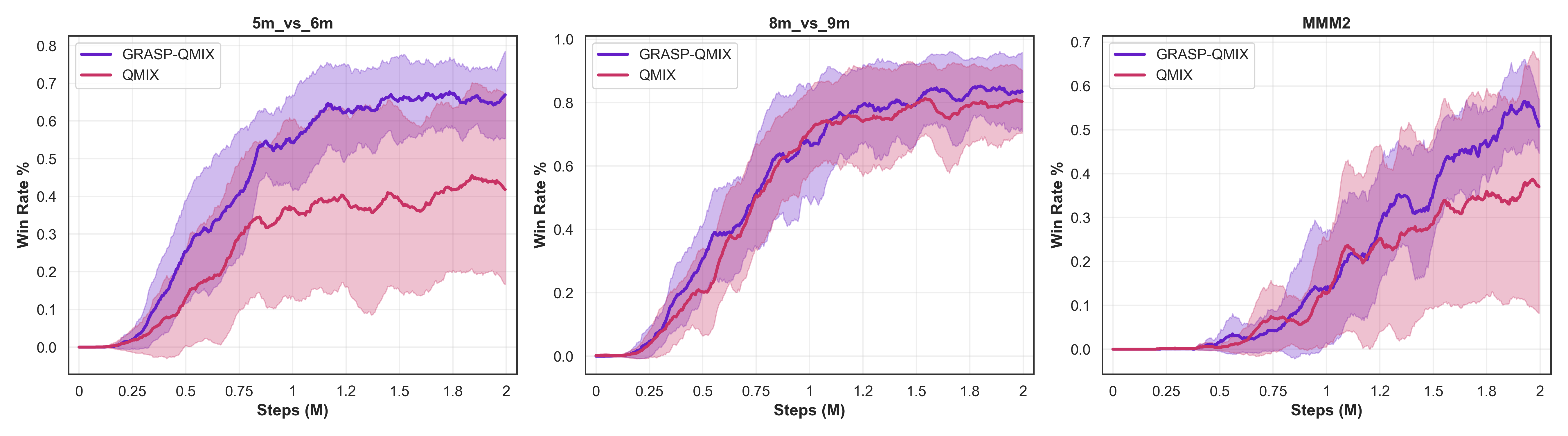}
    \caption{Comparison of win rate learning curves between GRASP-QMIX and QMIX across three scenarios in SMAC. Solid lines represent the average win rate from multiple runs, while shaded areas indicate standard deviation.}
    \label{fig:grasp_qmix_res}
\end{figure*}

To validate the portability of the GRASP framework, we ported our method to other approaches. Although GRASP aims to find the consensus gradient—the most suitable for on-policy methods—because according to~\ref{the1:u_for_pi}, the $g_i$ used to compute $u^*$ should be the gradient of the policy at the current time step. However, in off-policy methods, the buffer contains excessive traces from previous policies, causing the computed consensus gradient to be influenced by past policies and thus not the most precise. In this case, $u$ degenerates into a suboptimal consensus gradient. However, even though the trajectories originate from previous policies, the current policy is still utilized when computing the TD error. Thus, the suboptimal consensus gradient can still effectively optimize off-policy methods. Consequently, we combined GRASP with several typical off-policy methods, QMIX, to demonstrate the validity of this inference. 

Fig.~\ref{fig:grasp_qmix_res} demonstrates the results of GRASP combined with QMIX. As inferred from the preceding analysis, GRASP continues to enhance the effectiveness of QMIX, although the improvement is less pronounced than that observed when GRASP enhances MAPPO.

\subsection{Additional  Extensibility of GRASP}
\label{ap:MPE_2730}

\begin{figure}[ht!]
    \centering
    \includegraphics[width=0.4\linewidth]{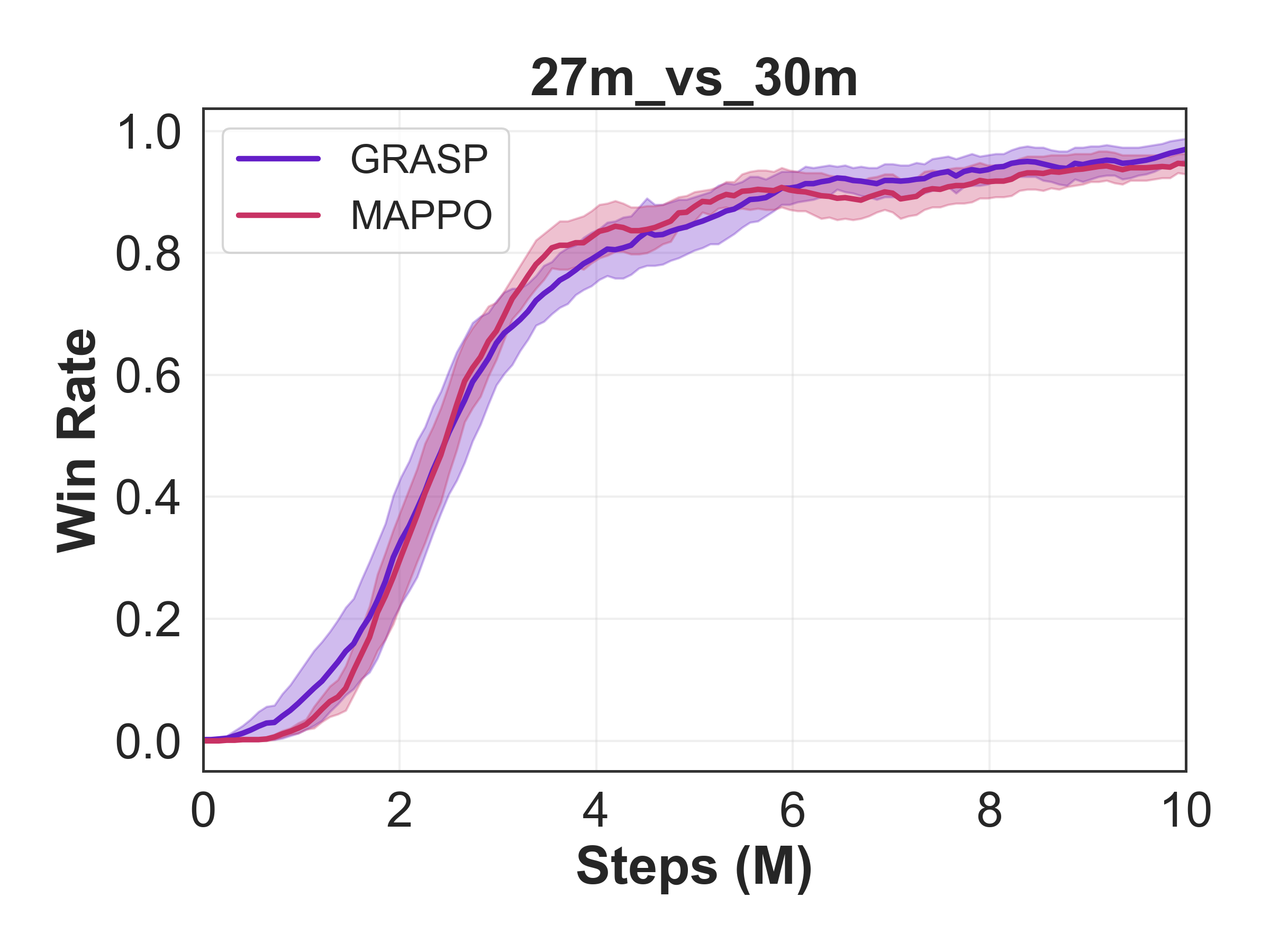}
    \caption{Comparison of win rate learning curves between GRASP and MAPPO in 27m\_vs\_30m scenario of MPE. Solid lines represent the average win rate from multiple runs, while shaded areas indicate standard deviation.}
    \label{fig:grasp_2730}
\end{figure}

To validate the scalability of the GRASP framework in large-scale agent collaboration tasks, we selected the Super Hard difficulty scenario 27m\_vs\_30m from the SMAC benchmark for experimentation. As the number of agents surges, the non-stationarity of the environment intensifies significantly, and coordinating conflicting policies among agents becomes increasingly challenging. Therefore, this scenario not only tests algorithmic performance but also pushes GRASP to its limits in extracting consensus within an extremely high-dimensional state space.

Fig.~\ref{fig:grasp_2730} presents the experimental results. In this highly challenging Super Hard scenario, MAPPO has demonstrated exceptional performance with a success rate of 94.68\%. Nevertheless, GRASP achieved a significant performance breakthrough, ultimately raising the success rate to 96.97\%. Notably, this improvement is far from negligible in the demanding 27m\_vs\_30m scenario. It signifies that GRASP reduces the task failure rate by nearly 43\% relative to the SOTA baseline. This demonstrates that the consensus mechanism of GRASP effectively corrects policy deviations, causing minor errors during large-scale agent battles.

\section{Algorithm}
\label{ap:alg}

\begin{algorithm}[ht!]
   \caption{GRASP (Shared Reward / Centralized Critic)}
   \label{alg:grasp}
\begin{algorithmic}[1]
   \STATE {\bfseries Input:} Initial policy parameters $\theta_i$ for agents $i \in \{1,\dots,N\}$, global value parameters $\phi$
   \STATE {\bfseries Hyperparameters:} Learning rate $\eta$, Clip $\epsilon$, GAE $\lambda$, Discount $\gamma$
   \STATE Initialize $\theta_{i, \text{old}} \leftarrow \theta_i$ for all $i$
   
   \FOR{iteration $k=1, 2, \dots$}
      \STATE \textit{// Phase 1: Data Collection \& Global GAE}
      \STATE Collect trajectory batch $\mathcal{D}$ by running $\pi_{\text{old}}$ for $T$ timesteps
      \STATE \textbf{Compute global value targets $\hat{R}_t$ and advantages $\hat{A}_t$ using GAE$(\gamma, \lambda)$ based on shared reward $r_t$ and $V(s_t; \phi)$}

      \STATE \textit{// Phase 2: Consensus Gradient Determination}
      \FOR{each agent $i \in \{1,\dots,N\}$}
         \STATE Compute local gradient on batch $\mathcal{D}$ using global $\hat{A}_t$:
         \STATE $g_i \leftarrow \mathbb{E}_{\mathcal{D}} \left[ \nabla_{\theta_i} \log \pi_i(a_{i,t}|o_{i,t}; \theta_{i,\text{old}}) \cdot \hat{A}_t \right]$
      \ENDFOR
      \STATE Solve QP for Pareto-optimal consensus direction:
      \STATE $u^* \leftarrow \operatorname*{argmin}_{u \in \operatorname{ConvexHull}(\{g_1, \dots, g_N\})} \|u\|^2$

      \STATE \textit{// Phase 3: Proximal Optimization Loop}
      \FOR{epoch $e=1$ to $K$}
         \STATE Shuffle $\mathcal{D}$ and partition into mini-batches
         \FOR{each mini-batch $\mathcal{B} \subset \mathcal{D}$}
            
            \STATE \textit{// 1. Global Critic Update}
            \STATE $\mathcal{L}(\phi) \leftarrow \frac{1}{|\mathcal{B}|} \sum_{t \in \mathcal{B}} (V(s_t; \phi) - \hat{R}_t)^2$
            \STATE $\phi \leftarrow \phi - \eta_{\text{critic}} \nabla_{\phi} \mathcal{L}(\phi)$

            \STATE \textit{// 2. Multi-Agent Actor Update}
            \FOR{each agent $i \in \{1,\dots,N\}$}
               \STATE Calculate ratio $\rho_t(\theta_i)$ and surrogate $J_i^{\text{PPO}}(\theta_i)$ using global $\hat{A}_t$
               \STATE Compute hybrid gradient: $\Delta \theta_i \leftarrow \nabla_{\theta_i} J_i^{\text{PPO}} + u^*$
               \STATE $\theta_i \leftarrow \theta_i + \eta \cdot \Delta \theta_i$
            \ENDFOR
         \ENDFOR
      \ENDFOR
      
      \STATE Update old policies: $\theta_{i, \text{old}} \leftarrow \theta_i$ for all $i$
   \ENDFOR
\end{algorithmic}
\end{algorithm}

\begin{algorithm}[h]
\caption{SolveConsensusQP: Solving Consensus Gradient via Quadratic Programming}
\label{alg:qp_solver}
\begin{algorithmic}[1]
\STATE \textbf{Input:} Set of local gradients $\mathcal{G} = \{g_1, g_2, \dots, g_N\}$
\STATE \textbf{Output:} Consensus gradient $u^*$, Mixing coefficients $\mathbf{c}^*$

\STATE \textbf{Construct QP Matrices:}
\STATE Let objective function be $f(\mathbf{c}) = \frac{1}{2} \| \sum_{i=1}^N c_i g_i \|_2^2 = \frac{1}{2} \mathbf{c}^\top \mathbf{P} \mathbf{c}$
\STATE Compute Gram Matrix $\mathbf{P} \in \mathbb{R}^{N \times N}$ where:
\STATE \quad $P_{ij} = g_i^\top g_j$ \COMMENT{Dot product of gradients}
\STATE Define Constraints:
\STATE \quad Equality: $\sum c_i = 1 \implies \mathbf{1}^\top \mathbf{c} = 1$
\STATE \quad Inequality: $c_i \ge 0 \implies \mathbf{I} \mathbf{c} \ge \mathbf{0}$

\STATE \textbf{Solve Standard QP:}
\STATE Find $\mathbf{c}^* = [c_1^*, \dots, c_N^*]^\top$ that minimizes:
\STATE \quad $\min_{\mathbf{c}} \frac{1}{2} \mathbf{c}^\top \mathbf{P} \mathbf{c}$
\STATE \quad s.t. $\mathbf{A}_{eq} \mathbf{c} = \mathbf{b}_{eq}$ \quad (where $\mathbf{A}_{eq} = \mathbf{1}^\top, \mathbf{b}_{eq} = 1$)
\STATE \quad \quad \ $\mathbf{A}_{in} \mathbf{c} \ge \mathbf{b}_{in}$ \quad (where $\mathbf{A}_{in} = \mathbf{I}, \mathbf{b}_{in} = \mathbf{0}$)

\STATE \textbf{Compute Consensus Gradient:}
\STATE $u^* \leftarrow \sum_{i=1}^N c_i^* g_i$

\STATE \textbf{return} $u^*, \mathbf{c}^*$
\end{algorithmic}
\end{algorithm}
\newpage

\section{Proof of the Theorem~\ref{th:eq_p_ex}}
\label{ap:Kakutani Fixed-Point Theorem}
To apply the Kakutani Fixed-Point Theorem, we define the GRASP update operator as:
\begin{equation}
    T_{\text{GRASP}}(\theta) = \left\{ \theta + \eta \left( g(\theta) +  u^*(\theta) \right) \right\}.
\end{equation}
Specifically, $u^*$ is computed by any operator satisfying the conditions in Condition~\ref{lem:Properties}. While these general properties suffice for an existence proof, to ensure theoretical consistency with our practical implementation, we employ the specific operator instantiation used in our algorithm for the derivation. The detailed definition and property verification of this operator are provided in Proposition~\ref{the2:KKT}. Therefore, $u^*(\theta) = \operatorname*{argmin}_{u \in \mathcal{H}(\theta)} \|u\|^2$, and $\mathcal{H}(\theta)$ denotes the convex hull of local gradients. We verify that $T_{\text{GRASP}}$ satisfies the three necessary conditions: (1) Non-emptiness, (2) Convexity, and (3) Closed Graph property.

1. Non-emptiness, we first establish that the consensus gradient $u^*$ always exists. The local gradients are defined as $g_i(\theta) = \mathbb{E} [ \nabla_{\theta_i} \log \pi_i \cdot A_i^{\boldsymbol{\pi}_{\theta}} ]$. The advantage $A_i^{\boldsymbol{\pi}_{\theta}}$ depends on the state value function $V^{\boldsymbol{\pi}_{\theta}}$, which is updated by minimizing the TD error:
\begin{equation}
    \delta_i(t) = r_i(t) + \gamma V_i(s_{t+1}; \phi) - V_i(s_t; \phi).
\end{equation}
Since the discount factor satisfies $\gamma \in [0, 1)$ and rewards are bounded ($|r_i| \le R_{\max}$), the Bellman operator acts as a contraction mapping. This guarantees that the value estimates are strictly bounded:
\begin{equation}
    |V^{\boldsymbol{\pi}_{\theta}}(s)| \le \frac{R_{\max}}{1 - \gamma} < \infty.
\end{equation}
Consequently, the local gradients $g_i(\theta)$ are finite for all $\theta \in \Theta$. The convex hull $\mathcal{H}(\theta)$ generated by these finite gradients is a compact set. By the Weierstrass Theorem, the continuous objective function $\|u\|^2$ achieves a minimum over this compact set, ensuring $u^*$ exists and $T_{\text{GRASP}}(\theta) \neq \emptyset$.

2. Convexity, we next examine the convexity of the image of the operator. The consensus gradient is the solution to the QP problem:
\begin{equation}
    \min_{u} \|u\|^2 \quad \text{s.t.} \quad u \in \mathcal{H}(\theta).
\end{equation}
Since the squared Euclidean norm $f(u) = \|u\|^2$ is a strictly convex function, and the constraint set $\mathcal{H}(\theta)$ is convex, the minimizer $u^*$ is unique. A set containing a single element (a singleton) is trivially convex. Thus, for any $\theta$, the set $T_{\text{GRASP}}(\theta)$ is convex.

3. Closed Graph, finally, we demonstrate that the operator has a closed graph. Given that $T_{\text{GRASP}}$ is single-valued, this is equivalent to proving continuity. Continuity of Constraints: Since $\pi_{\theta}$ is continuously differentiable, the mapping from parameters $\theta$ to the set of local gradients $\{g_i(\theta)\}$ is continuous. Consequently, the correspondence $\theta \mapsto \mathcal{H}(\theta)$ is continuous. Berge's Maximum Theorem: For the optimization problem $\min \|u\|^2$ s.t. $u \in \mathcal{H}(\theta)$, since the objective is continuous and the constraint correspondence is continuous, the theorem states that the solution mapping is upper hemi-continuous. Continuity of Solution: Coupled with the uniqueness of $u^*$, upper hemi-continuity implies that $u^*(\theta)$ is a continuous function of $\theta$. Since the update rule is a linear combination of continuous functions:
\begin{equation}
    \theta' = \theta + \eta (g(\theta) +  u^*(\theta)).
\end{equation}
In summary, the operator $T_{\text{GRASP}}$ is continuous.

The operator $T_{\text{GRASP}}$ satisfies all conditions of the Kakutani Fixed-Point Theorem. Therefore, there exists a fixed point $\theta^* \in \Theta$ such that $\theta^* = T_{\text{GRASP}}(\theta^*)$, proving the reachability of the generalized Bellman equilibrium.

Since $\theta$ is nonempty, compact, and convex, and $\Phi$ satisfies the required conditions, Kakutani Fixed-Point Theorem guarantees the existence of a fixed point $\theta^* \in \theta$ such that $\theta^* \in \Phi(\theta^*)$. By the definition of $\Phi$, this implies:
\begin{equation}
    \theta^* = \theta^* + \eta v(\theta^*) \implies v(\theta^*) = \boldsymbol{0}.
\end{equation}
This implies that for all $i$, $v_i(\theta^*) = g_i(\theta^*) + u^*(\theta^*) = \boldsymbol{0}$, i.e., $g_i(\theta^*) = -u^*(\theta^*)$. Substituting this relation into the definition of $u^*$ yields: $u^*(\theta^*) = \sum c_i^* (- u^*(\theta^*)) = -(\sum c_i^*) u^*(\theta^*) = -u^*(\theta^*)$, which holds only when $u^*(\theta^*) = \boldsymbol{0}$. Therefore, this fixed point is a \textbf{consensus equilibrium point} with zero consensus gradient.

\section{Proof of Value Function Convergence and Monotonic Improvement}
\label{ap:value_convergence}

In this section, we provide the detailed proof for the monotonic policy improvement and value function convergence under the GRASP framework, as outlined in~\cref{method:value_function_proof}.

\textbf{Objective Improvement via Taylor Expansion and Convex Geometry}: We first analyze the change in the objective function $J(\theta)$ induced by the parameter update. According to the definition of the First-Order Taylor Expansion, for a smooth function $J(\theta)$, the value at the updated parameter $\theta_{\text{new}} = \theta_{\text{old}} + \Delta \theta$ can be approximated as:
\begin{equation}
        J(\theta_{\text{new}}) =  J(\theta_{\text{old}}) + 
         \sum_{i=1}^N \langle \nabla_{\theta_i} J(\theta), \Delta \theta_i \rangle + o(\|\Delta \theta\|^2).
\end{equation}
Substituting the GRASP update rule $\Delta \theta_i = \eta (g_i +  u^*)$ into the linear term, the increment $\Delta J$ is given by:
\begin{equation}
\Delta J \approx \sum_{i=1}^N \langle g_i, \eta (g_i +  u^*) \rangle 
= \eta \sum_{i=1}^N \|g_i\|^2 + \eta  \sum_{i=1}^N \langle g_i, u^* \rangle.
\end{equation}
To analyze the sign of the second term, we invoke the fundamental property of the minimum-norm point in a convex hull, as established in the Frank-Wolfe algorithm and Multi-Gradient Descent literature~\cite{convex_desideri2012multiple,convex_frank1956algorithm}.

Since $\eta > 0$, the GRASP update guarantees a strictly non-negative improvement in the joint objective, surpassing the gain of standard independent updates by the margin of the consensus term $\eta  N \|u^*\|^2$.

\textbf{Value Function Convergence via Policy Improvement Theorem:} Having established that the joint objective (expected return) increases, we now demonstrate that this leads to a uniform improvement in the value function $V(s)$. We rely on the Policy Improvement Theorem~\cite{bertsekas2019reinforcement,sutton1998reinforcement}.

Since the update direction aligns with the ascent direction of the objective $J(\theta)$, the new policy $\pi_{\text{new}}$ improves the expected action value over the old policy $\pi_{\text{old}}$ locally:
\begin{equation}
    Q^{\pi_{\text{old}}}(s, \pi_{\text{new}}(s)) \ge V^{\pi_{\text{old}}}(s), \quad \forall s.
\end{equation}
To prove that this local improvement translates to the global value function $V^{\pi_{\text{new}}}(s)$, we perform the standard recursive expansion (telescoping argument). Starting from the definition of the value function:
\begin{equation}
\begin{aligned}
V^{\pi_{\text{old}}}(s) &\le Q^{\pi_{\text{old}}}(s, \pi_{\text{new}}(s)) \\
&= \mathbb{E}_{a \sim \pi_{\text{new}}, s' \sim \mathcal{P}} \left[ r(s, a) + \gamma V^{\pi_{\text{old}}}(s') \right] \\
&\le \mathbb{E}_{a \sim \pi_{\text{new}}, s' \sim \mathcal{P}} \left[ r(s, a) + \gamma Q^{\pi_{\text{old}}}(s', \pi_{\text{new}}(s')) \right] \quad  \\
&= \mathbb{E}_{\tau \sim \pi_{\text{new}}} \left[ r_t + \gamma r_{t+1} + \gamma^2 V^{\pi_{\text{old}}}(s_{t+2}) \right] \\
&\le \mathbb{E}_{\tau \sim \pi_{\text{new}}} \left[ \sum_{k=0}^{\infty} \gamma^k r_{t+k} \right] \\
&= V^{\pi_{\text{new}}}(s).
\end{aligned}
\end{equation}
The derivation confirms that $V^{\pi_{\text{new}}}(s) \ge V^{\pi_{\text{old}}}(s)$ for all $s$. Consequently, the standard Critic, which minimizes the Temporal Difference (TD) error against the returns generated by $\pi_{\text{new}}$, will naturally converge to the higher value estimates $V^{\pi_{\text{new}}}(s)$.

\section{Convergence of Value Function Estimates under GRASP}
\label{ap:convergence_value_function}

\begin{theorem}
     Let $\mathcal{T}^{\pi}$ be the standard Bellman evaluation operator for a policy $\pi$. Under the GRASP framework, the Critic update converges to a unique fixed point $V^{\pi_{\text{new}}}$, and the sequence of value functions generated by iterative GRASP updates satisfies the monotonicity property: $V^{\pi_{k+1}}(s) \ge V^{\pi_{k}}(s)$.
\end{theorem}

\begin{proof}
1. Contraction Property of the Standard Bellman Operator. First, we establish that for any fixed policy $\pi$ (including the GRASP-updated policy $\pi_{\text{new}}$), the Critic update is a contraction mapping. The standard Bellman operator $\mathcal{T}^{\pi}$ is defined as:
\begin{equation}
    (\mathcal{T}^{\pi} V)(s) = r(s, \pi(s)) + \gamma \sum_{s'} \mathcal{P}(s'|s, \pi(s)) V(s').
\end{equation}
For any two value functions $V$ and $V'$, and for any state $s$:
\begin{equation}
    \begin{aligned}
|(\mathcal{T}^{\pi} V)(s) - (\mathcal{T}^{\pi} V')(s)| &= \left| \gamma \sum_{s'} \mathcal{P}(s'|s, \pi(s)) (V(s') - V'(s')) \right| \\
&\le \gamma \sum_{s'} \mathcal{P}(s'|s, \pi(s)) \|V - V'\|_{\infty} \\
&= \gamma \|V - V'\|_{\infty}.
\end{aligned}
\end{equation}
Since the discount factor $\gamma \in [0, 1)$, $\mathcal{T}^{\pi}$ is a $\gamma$-contraction under the max-norm. By the Banach Fixed-Point Theorem, the Critic update iteratively converges to a unique fixed point $V^{\pi}$, which is the true value of policy $\pi$.

Shift of the Fixed Point via GRASP. The core contribution of GRASP is not modifying the operator $\mathcal{T}$, but shifting the fixed point to a superior position.Let $\pi_k$ be the policy at iteration $k$, and $\pi_{k+1}$ be the policy updated via GRASP: $\pi_{k+1} \leftarrow \pi_k + \eta(g +  u^*)$. From Proposition~\ref{def:monotomic_policy_improvent} (Monotonic Policy Improvement), we established that $\Delta J \ge 0$, which implies local improvement in the action-value function:
\begin{equation}
    Q^{\pi_k}(s, \pi_{k+1}(s)) \ge V^{\pi_k}(s).
\end{equation}
Consider the Bellman operator for the new policy, $\mathcal{T}^{\pi_{k+1}}$. Applying this operator to the old value function $V^{\pi_k}$:
\begin{equation}
    \begin{aligned}
(\mathcal{T}^{\pi_{k+1}} V^{\pi_k})(s) &= \mathbb{E}_{a \sim \pi_{k+1}} [r + \gamma V^{\pi_k}(s')] \\
&= Q^{\pi_k}(s, \pi_{k+1}(s)) \\
&\ge V^{\pi_k}(s) \ .
\end{aligned}
\end{equation}

This inequality $(\mathcal{T}^{\pi_{k+1}} V^{\pi_k}) \ge V^{\pi_k}$ is the condition for Monotonicity in Policy Iteration. By repeatedly applying the operator $\mathcal{T}^{\pi_{k+1}}$, the value estimates strictly increase (or stay constant at optimality) until converging to the new fixed point:
\begin{equation}
    V^{\pi_{k+1}} = \lim_{n \to \infty} (\mathcal{T}^{\pi_{k+1}})^n V^{\pi_k} \ge V^{\pi_k}.
\end{equation}

Conclusion: The GRASP framework ensures convergence in two steps: Actor Step: The gradient realignment guarantees a transition from $\pi_k$ to a strictly better policy $\pi_{k+1}$. Critic Step: The standard Bellman update, acting as a contraction operator, converges to the true value $V^{\pi_{k+1}}$. Due to the monotonic nature of the policy update, the sequence of fixed points is non-decreasing: $V^{\pi_0} \le V^{\pi_1} \le \dots \le V^*$, ensuring stable convergence toward the optimal cooperative equilibrium. $\hfill \blacksquare$

Remark on Implicit Convergence: Unlike methods that artificially alter the target of Critic (e.g., modifying rewards or adding regularization terms), GRASP relies on implicit synchronization. The Consensus Module actively realigns the Actor to produce higher-quality trajectories (higher returns $R$). The Critic, observing these improved returns, naturally adjusts its value estimates $V(s)$ upwards through the standard Bellman contraction. This separation of concerns, Actor handling the Direction of Improvement and Critic handling the Evaluation of Reality, ensures that the algorithm retains the stability guarantees of standard PPO while enjoying the accelerated convergence provided by consensus.

\end{proof}
\section{Proof that QP satisfies the properties of the Condition~\ref{lem:Properties}}
\label{apd:kkt_proof}
1) Constructing the Lagrange Function
    
    Let the objective function be the squared distance from the origin to the convex combination point. Introduce the Lagrange multiplier $\lambda$ (corresponding to the equality constraint $\sum c_i = 1$) and $\boldsymbol{\mu} = [\mu_1, \dots, \mu_N]^\top$ (corresponding to the inequality constraint $c_i \ge 0$):
\begin{equation}
    \mathcal{L}(\boldsymbol{c}, \lambda, \boldsymbol{\mu}) = \frac{1}{2} \left\| \sum_{i=1}^N c_i g_i \right\|_2^2 - \lambda \left(\sum_{i=1}^N c_i - 1\right) - \sum_{i=1}^N \mu_i c_i.
\end{equation}
2) Using the KKT conditions

At the optimal solution $\boldsymbol{c}^*$, the partial derivative with respect to each component $c_j$ must be zero:
\begin{equation}
    \frac{\partial \mathcal{L}}{\partial c_j} = \left(\sum_{i=1}^N c_i^* g_i\right)^\top g_j - \lambda - \mu_j = 0.
\end{equation}
Since $u^* = \sum_{i=1}^N c_i^* g_i$, the above equation can be rewritten as:
\begin{equation}
\label{eq:kkt_}
    u^{*T} g_j = \lambda + \mu_j \quad \forall j \in \{1, \dots, N\} .
\end{equation}
3) Solve for the dual variable $\lambda$

Multiply both sides of (Eq.~\ref{eq:kkt_}) by $c_j^*$ and sum over all $j$:
\begin{equation}
    \sum_{j=1}^N c_j^* (u^{*T} g_j) = \sum_{j=1}^N c_j^* \lambda + \sum_{j=1}^N c_j^* \mu_j.
\end{equation}
Left-hand side: $\sum c_j^* (u^{*T} g_j) = u^{*T} (\sum c_j^* g_j) = u^{*T} u^* = \|u^*\|_2^2$.

Right-hand side: Since $\sum c_j^* = 1$ (Proposition~\ref{the2:KKT}) and $c_j^* \mu_j = 0$ (complementary slackness), the right-hand side simplifies to $\lambda \cdot 1 + 0 = \lambda$.

Thus, we derive a crucial geometric property: the dual variable $\lambda$ precisely equals the square of the consensus gradient norm.
\begin{equation}
    \lambda = \|u^*\|_2^2.
\end{equation}
4) Verifying the Pareto Property

Substituting $\lambda = \|u^*\|_2^2$ back into (Eq.~\ref{eq:kkt_}):
\begin{equation}
    g_j^\top u^* = \|u^*\|_2^2 + \mu_j.
\end{equation}
By the dual feasibility constraint, the multiplier $\mu_j$ must be nonnegative ($\mu_j \ge 0$). Therefore:
\begin{equation}
    g_j^\top u^* \ge \|u^*\|_2^2.
\end{equation}
Since the square of a norm is always non-negative ($|u^*\|_2^2 \ge 0$), we ultimately prove that:
\begin{equation}
    g_j^\top u^* \ge 0.
\end{equation}
5) Generalized Bellman Equilibrium

When training approaches stability, if the system reaches a fixed point, according to Theorem~\ref {th:eq_p_ex}:
The update step becomes 0, i.e., $v_i(\theta^*) = 0$. This implies $g_i(\theta^*) = -u^*(\theta^*)$. Substituting this into the definition of $u^*$ yields: $u^* = \sum c_i^* (-u^*) = -(\sum c_i^*) u^* = -u^*$. The only possible case where this equation holds is $u^* = 0$*.

\section{Off-policy Methods within GRASP}
\label{ap:off-policy_grasp}
\subsection{Geometrically Aligned Update Rules}

To ensure stability in off-policy methods, we introduce a Geometrically Aligned Factor $\Gamma_i$ that imposes an intrinsic trust-region constraint directly on the update magnitude. This factor serves as the primary stability mechanism by strictly bounding policy shifts in single-step update paradigms. This ensures that the active gradient realignment does not violate the monotonicity assumption of the policy improvement theorem, effectively marrying the stability of trust region methods with the coordination efficiency of game-theoretic consensus. Instead of relying on heuristic step-size clipping, we reconstruct the proactive update vector $v_i$ using a geometrically-driven scaling mechanism defined as follows:
\begin{equation}
v_i = \Gamma_i (g_i + u^*),
\end{equation}
where $\Gamma_i$ is the \textbf{Geometrically Aligned Factor}. Grounded in trust-region principles, $\Gamma_i$ adaptively scales the update magnitude based on the resonance between individual and collective gradients:
\begin{equation}
\Gamma_i = \frac{\|u^*\|^2 + \langle g_i, u^* \rangle}{\|g_i\|^2 + \|u^*\|^2}.
\end{equation}

This formulation offers several theoretical advantages:
\begin{itemize}
    \item \textbf{Intrinsic Trust-Region Constraint}: The factor $\Gamma_i$ effectively acts as a dynamic damping term. As proven in Theorem~\ref{th:eq_p_ex}, $\langle g_i, u^* \rangle \ge 0$, which ensures $\Gamma_i \in [0, 1]$. This naturally keeps the update within a safe proximity of the current policy parameters.
    \item \textbf{Parameter-Free Adaptation}: Unlike traditional clipping methods that require a predefined threshold, the magnitude of $v_i$ is entirely determined by the relative geometry of the gradient space, enhancing the robustness of the optimization process.
\end{itemize}

\subsection{Proof of Stability and Consistency}

Next, we prove that under these derived constraints, the GRASP update is stable.

\begin{proposition}[Stability and Consistency]
For each agent $i$, the constrained GRASP update direction $\tilde{v}_i$ satisfies two key properties:
\begin{enumerate}
    \item \textbf{Individual Improvement: } $g_i^\top \tilde{v}_i \ge 0$. Ensure that no agent is harmed during the learning process. 
    \item \textbf{Collective Coherence: } $(u^*)^\top \tilde{v}_i \ge 0$. Ensure that individual updates never deviate from the consensus update direction of the team.
\end{enumerate}
\end{proposition}

\begin{proof}
\text{Property 1 (Individual Improvement):}
Substituting the resonance update rule $v_i = \Gamma_i (g_i + u^*)$ into the inner product with the individual gradient $g_i$:
\begin{equation}
\langle g_i, v_i \rangle = \Gamma_i \langle g_i, g_i + u^* \rangle = \Gamma_i (|g_i|^2 + \langle g_i, u^* \rangle).
\end{equation}
According to the Incentive Compatibility proven in Proposition~\ref{the2:KKT}, the consensus direction $u^*$ satisfies $\langle g_i, u^* \rangle \ge 0$ for all $i$. Furthermore, the Geometrically Aligned Factor $\Gamma_i$, derived from trust-region considerations, is defined as:
\begin{equation}
\Gamma_i = \frac{|u^*|^2 + < g_i, u^* >}{|g_i|^2 + |u^*|^2}.
\end{equation}
Since $|u^*|^2 \ge 0$ and $\langle g_i, u^* \rangle \ge 0$, the numerator of $\Gamma_i$ is non-negative, and the denominator is strictly positive (unless $g_i = u^* = 0$), ensuring $\Gamma_i \ge 0$. Consequently:
\begin{equation}
\langle g_i, v_i \rangle = \underbrace{\Gamma_i}_{\ge 0} (\underbrace{|g_i|^2}_{\ge 0} + \underbrace{\langle g_i, u^* \rangle}_{\ge 0}) \ge 0.
\end{equation}
This confirms that the update direction $v_i$ is a descent direction for the individual cost (or an ascent direction for the reward), maintaining the individual rationality within the trust region.

\textit{Property 2 (Collective Consistency):}
Similarly, for the consensus direction $u^*$:
\begin{equation}
\langle u^*, v_i \rangle = \Gamma_i \langle u^*, g_i + u^* \rangle = \Gamma_i (|u^*|^2 + \langle g_i, u^* \rangle).
\end{equation}
By the same non-negativity of $\Gamma_i$ and $\langle g_i, u^* \rangle$, we have:
\begin{equation}
\langle u^*, v_i \rangle \ge 0.
\end{equation}
This establishes that the individual update $v_i$ never opposes the collective intent $u^*$, ensuring the coherence of the multi-agent system during the optimization process.

\end{proof}

\begin{remark}
    The above process demonstrates that GRASP successfully decouples individual rationality from collective stability. Agents are free to follow their own gradients, but these gradients are regularized by a consensus term. This ensures alignment with team interests and prevents the disruptive oscillations commonly observed among independent learners.
\end{remark}

\newpage
\section{Parameters for the Comparative Experiment}
\label{ap:settingup}
In this section, we provide detailed hyperparameter settings for each algorithm and experimental scenario. 
~\cref{tab:hyper-comparison} presents a comparison of the general hyperparameters among the proposed GRASP, $\text{MA}^2\text{E}$, and baseline algorithms (MAPPO and HAPPO). It is observed that most baselines maintain consistency in the learning rate ($5 \times 10^{-4}$) and the choice of optimizer to ensure a fair comparison. Notably, $\text{MA}^2\text{E}$ employs a larger buffer size (5000) while integrating a QMIX mixing network, reflecting its unique mechanism for handling long-sequence data.

~\cref{tab:ma2e} further details the specific internal parameters of the $\text{MA}^2\text{E}$ model. Based on the Transformer architecture, the model is configured with 3 encoder layers and 2 decoder layers, supported by a 4-head attention mechanism. Furthermore, the settings for fine-tuning steps (500) and the pre-training threshold (0.015) are designed to balance representation learning on large-scale data with rapid transferability in specific scenarios.

\begin{table}[ht]
\caption{Hyperparameter settings for GRASP, MAPPO, HAPPO, and $\text{MA}^2\text{E}$.}
\label{tab:hyper-comparison}
\centering
\begin{small}
\begin{tabular}{lcccccc}
\toprule
Hyperparameter & GRASP & MAPPO & HAPPO & $\text{MA}^2\text{E}$  &GRASP+QMIX &QMIX\\
\midrule
Learning Rate & $5 \times 10^{-4}$ & $5 \times 10^{-4}$ & $5 \times 10^{-4}$ & $1 \times 10^{-3}$ & $1 \times 10^{-3}$ & $1 \times 10^{-3}$\\
Optimizer & Adam & Adam & Adam & Adam & Adam& Adam\\
PPO Epochs & 5 -- 10 & 5 -- 10 & 5 & N/A & N/A & N/A \\
Mini-batch Size & 1 -- 2 & 1 -- 2 & 1 & 32& 32& 32 \\
Gamma ($\gamma$) & 0.99 & 0.99 & 0.99 & 0.99 & 0.99 & 0.99 \\
Buffer size & 64 & 64 &64 &5000&5000&5000\\
Policy Update (episode) & 200 & 200 & 200 &1000&200 &200\\ 
ValueNorm & True & True & True & N/A& N/A& N/A \\
Mixing Network & N/A & N/A & N/A & QMIX& QMIX& QMIX \\
Consensus Operator  & QP & None & None & None & None& None\\
\bottomrule
\end{tabular}
\end{small}
\end{table}
We employ a Recurrent Neural Network (RNN) architecture for the policy of the agent and value approximation. Specifically, the network consists of three components: (1) a linear input layer that encodes the raw observation into a latent embedding; (2) a hidden layer utilizing a Gated Recurrent Unit (GRU) with a hidden state size of 64, followed by a ReLU activation function to capture temporal dependencies; and (3) a linear output layer that projects the hidden features to the action space. The detailed network hyperparameters are listed in Table~\ref{tab:agent_arch}.
\begin{table}[h!]
    \centering
    \caption{Agent Network Architecture Hyperparameters.}
    \label{tab:agent_arch}
    \begin{tabular}{llcc}
        \toprule
        \textbf{Component} & \textbf{Layer Type} & \textbf{Hidden Units} & \textbf{Activation} \\
        \midrule
        Input Encoder & Full Connect  Linear  & 64 & - \\
        Recurrent Layer & GRU & 64 & ReLU \\
        Output Head & Full Connect Linear  & $A_{\text{dim}}$ & - \\
        \bottomrule
    \end{tabular}
    \vspace{0.2cm} \\
    \footnotesize{\textit{Note:} $A_{\text{dim}}$ denotes the dimension of the action space. The input dimension corresponds to the individual partial observation dimension $O_{\text{dim}}$.}
\end{table}
~\cref{tab:critic_arch} and ~\cref{tab:mixer_arch} present the detailed hyperparameter configurations for the Centralized Critic and the Mixing Network, respectively. The Centralized Critic (Table~\ref{tab:critic_arch}) utilizes a Multi-Layer Perceptron (MLP) architecture to estimate the global state value, comprising two hidden layers with 64 units each and ReLU activation. The Mixing Network (Table~\ref{tab:mixer_arch}), following the QMIX paradigm, employs a hypernetwork-based structure. It maps the global state to the weights of the mixing layers, using a mixing embedding dimension of 32 to enforce the monotonicity constraint between individual agent utilities and the joint value.

Furthermore, we would like to clarify the network architecture design in GRASP. The policy networks of agents utilize a hybrid architecture: the feature representation backbone is shared across agents, while the specific policy heads remain independent. This design improves sample efficiency through collective experience while simplifying the computation of individual policy gradients $g_i$. However, GRASP is equally applicable to fully parameter-sharing methods. In a fully shared architecture, individual gradients can still be derived by leveraging the correlation of the input data. Specifically, by setting the input tensors to require gradients, we can compute the independent gradient of the total loss with respect to each agent's specific input slice ($g_i = \nabla_{\text{input}_i} \mathcal{L}$). While theoretically sound, this approach is computationally expensive and complex to implement. As shown in Listing~\ref{lst:grad_code}, it requires manually managing leaf tensors and retaining computation graphs during backpropagation for each agent.

\begin{lstlisting}[language=Python, caption={Pseudo-code for computing independent gradients in a fully shared architecture.}, label={lst:grad_code}]
inputs.requires_grad = True
g_i = torch.autograd.grad(
        outputs=loss, 
        inputs=inputs[agent_id], 
        retain_graph=True  
    )[0]
\end{lstlisting}

This repeated graph retention (\texttt{retain\_graph=True}) significantly inflates memory overhead and training time. Therefore, we adopted the hybrid architecture to perfectly balance computational efficiency and consensus performance.

\begin{table}[h!]
    \centering
    \caption{Centralized Critic Network Architecture.}
    \label{tab:critic_arch}
    \begin{tabular}{llcc}
        \toprule
        \textbf{Layer} & \textbf{Type} & \textbf{Hidden Units} & \textbf{Activation} \\
        \midrule
        Input & Full Connect Linear  & 64 & ReLU \\
        Hidden 1 & Full Connect Linear  & 64 & ReLU \\
        Hidden 2 & Full Connect Linear  & 64 & ReLU \\
        Output & Full Connect Linear  & 1 & - \\
        \bottomrule
    \end{tabular}
    \vspace{0.1cm} \\
    \footnotesize{\textit{Note:} The input dimension corresponds to the global state dimension $S_{\text{dim}}$.}
\end{table}

\begin{table}[h!]
    \centering
    \caption{Mixing Network and Hypernetwork Architecture.}
    \label{tab:mixer_arch}
    \begin{tabular}{llcc}
        \toprule
        \textbf{Component} & \textbf{Layer Type} & \textbf{Output Dim} & \textbf{Activation} \\
        \midrule
        \multicolumn{4}{l}{\textit{Mixing Network (Inputs: $N$ Agent Q-values)}} \\
        \cmidrule(lr){1-4}
        Mixing Layer 1 & Full Connect Linear & 32 & ELU \\
        Mixing Output & Full Connect Linear & 1 & - \\
        \midrule
        \multicolumn{4}{l}{\textit{Hypernetworks (Input: Global State $S$)}} \\
        \cmidrule(lr){1-4}
        Hyper-Input & Full Connect Linear & 64 & ReLU \\
        Hyper-Weight 1 & Full Connect Linear  & $N \times 32$ & Abs \\
        Hyper-Weight 2 & Full Connect Linear  & $32 \times 1$ & Abs \\
        Hyper-Bias & Full Connect Linear  & 32 / 1 & - \\
        \bottomrule
    \end{tabular}
    \vspace{0.1cm} \\
    \footnotesize{\textit{Note:} $N$ denotes the number of agents. Weights are generated dynamically by the Hypernetworks based on the state $S$.}
\end{table}

\begin{table}[h!]
\centering
\caption{The hyperparameter settings for $\text{MA}^2\text{E}$}
\label{tab:ma2e}
\begin{tabular}{ll}
\toprule
Hyperparameters & Value \\
\midrule
Batch size & 32 \\
Input embedding & 24 \\
The number of heads & 4 \\
The number of encoder layers & 3 \\
The number of decoder layers & 2 \\
Steps for fine-tuning & 500 \\
Pretraining threshold & 0.015 \\
\bottomrule
\end{tabular}
\end{table}

\begin{table*}[h!]
\centering
\caption{Characteristics of the Selected Google Research Football Scenarios. These scenarios were chosen to evaluate the GRASP framework across varying levels of difficulty, agent scale, and tactical requirements.}
\label{tab:grf_scenarios}
\resizebox{\textwidth}{!}{%
\begin{tabular}{@{}lcccc@{}}
\toprule
\textbf{Scenario Name} & \textbf{Scale} & \textbf{Description} & \textbf{Key Challenges} & \textbf{Testing Focus} \\ \midrule
\textbf{Academy Counterattack Hard} & \begin{tabular}[c]{@{}c@{}}Small\\ (4 Agents)\end{tabular} & \begin{tabular}[c]{@{}c@{}}Rapid transition from defense\\ to offense at high speed.\end{tabular} & \begin{tabular}[c]{@{}c@{}}High Dynamics,\\ Split-second Decision Making\end{tabular} & \begin{tabular}[c]{@{}c@{}}Responsiveness,\\ Dynamic Control\end{tabular} \\ \midrule
\textbf{Academy 3 Vs 1 With Keeper} & \begin{tabular}[c]{@{}c@{}}Small\\ (3 Agents)\end{tabular} & \begin{tabular}[c]{@{}c@{}}Exploiting numerical advantage\\ to dismantle defense.\end{tabular} & \begin{tabular}[c]{@{}c@{}}Local Coordination,\\ Spatial Reasoning\end{tabular} & \begin{tabular}[c]{@{}c@{}}Tactical Cooperation,\\ Precise Passing\end{tabular} \\ \midrule
\textbf{Academy Corner} & \begin{tabular}[c]{@{}c@{}}Large\\ (11 Agents)\end{tabular} & \begin{tabular}[c]{@{}c@{}}Full-team corner kick situation\\ with crowded penalty area.\end{tabular} & \begin{tabular}[c]{@{}c@{}}Massive State Space,\\ Occlusion \& Collisions\end{tabular} & \begin{tabular}[c]{@{}c@{}}Scalability,\\ Consensus in Noise\end{tabular} \\ \bottomrule
\end{tabular}%
}
\end{table*}

\begin{table}[h!]
\caption{Detailed configurations for SMAC and SMACv2 scenarios.}
\label{tab:smac-details}
\centering
\begin{small}
\begin{tabular}{lccccc}
\toprule
Scenario & Ally Units & Enemy Units & Difficulty & Step Limit \\
\midrule
\textit{SMAC} & & & & \\
5m\_vs\_6m & 5 Marines & 6 Marines & Hard & 70 \\
MMM2 & 10 Units\footnotemark[1] & 12 Units\footnotemark[2] & Super Hard & 180  \\
3s5z\_vs\_3s6z & 8 Units\footnotemark[3] & 9 Units\footnotemark[4] & Hard & 170 \\
27m\_vs\_30m & 27 Marines & 30 Marines & Super Hard & 180 \\
8m\_vs\_9m & 8 Marines & 9 Marines & Hard &  120 \\
\midrule
\textit{SMACv2} & & & & \\
Terran 10v10 & 10 Units & 10 Units & --- & 200  \\
Protoss 10v10 & 10 Units & 10 Units & --- & 200 \\
Zerg 5v5 & 5 Units & 5 Units & --- & 200 \\
\bottomrule
\end{tabular}
\end{small}
\footnotetext[1]{1 Medivac, 2 Marauders, 7 Marines}
\footnotetext[2]{1 Medivac, 3 Marauders, 8 Marines}
\footnotetext[3]{3 Stalkers, 5 Zealots}
\footnotetext[4]{3 Stalkers, 6 Zealots}
\end{table}

\begin{table}[h!]
\caption{Configurations for Google Football and MPE scenarios.}
\label{tab:other-env-details}
\centering
\begin{small}
\begin{tabular}{lccc}
\toprule
Scenario & Ally Agents & Step Limit & Reward Type \\
\midrule
\textit{Google Football} & & & \\
Counterattack Hard & 4 & 400 & Team collaboration provides \\
3 Vs 1 with Keeper & 3 & 400 & Team collaboration provides \\
Academy Corner & 10 & 400 & Team collaboration provides \\
\midrule
\textit{MPE} & & & \\
Simple Spread & 3 & 25 & Individual rewards are consolidated into team rewards \\
\bottomrule
\end{tabular}
\end{small}
\end{table}

~\cref{tab:grf_scenarios} enumerates the characteristics of the selected GRF scenarios. We have carefully chosen tasks of varying scales: ranging from Academy 3 Vs 1 With Keeper (3 agents), which focuses on local coordination and precise passing, to Academy Counterattack Hard (4 agents), which involves high-speed dynamic decision-making, and finally Academy Corner (11 agents), which tests the scalability of algorithms in a massive state space.

As shown in~\cref{tab:smac-details}, our evaluation covers various classic maps ranging from Hard to Super Hard difficulty. Each scenario is governed by a strict Step Limit (e.g., 70 steps for 5m\_vs\_6m and 180 steps for MMM2), requiring agents to complete the objective of defeating enemies within a constrained time budget.

~\cref{tab:other-env-details} summarizes the configurations for GRF and the MPE. In GRF, the step limit for all scenarios is uniformly set to 400 steps with a team-based reward mechanism. In contrast, the Simple Spread task in MPE has a limit of only 25 steps, focusing on the ability of agents to achieve collaborative coverage within an extremely short horizon.

\begin{table}[h!]
\centering
\caption{Hardware and Software Configurations for the Training Environment.}
\label{tab:hardware-config}
\begin{tabular}{ll}
\toprule
Component & Configuration Detail \\
\midrule
CPU Model & 13th Gen Intel(R) Core(TM) i5-13400 \\
CPU Architecture & x86\_64 (10 Cores, 16 Threads) \\
CPU Max Frequency & 4.60 GHz \\
GPU Model & NVIDIA GeForce RTX 4090 \\
GPU Memory & 24564 MiB (24 GB GDDR6X) \\
Driver Version & 535.183.01 \\
CUDA Version & 12.2 \\
\bottomrule
\end{tabular}
\end{table}

To evaluate the algorithmic performance in a high-performance computing environment, all experiments were conducted on a hardware platform equipped with professional-grade GPUs. ~\cref{tab:hardware-config} details the hardware and software configurations used in this study. We utilized a 13th Gen Intel Core i5-13400 processor as the central computing unit, whose architecture of 10 cores and 16 threads efficiently handles parallel sampling tasks in multi-agent environments. The experiments leveraged an NVIDIA GeForce RTX 4090 GPU, featuring 24 GB of VRAM and 16,384 CUDA cores, providing robust computational support for the rapid iteration of deep neural networks. Furthermore, the system operates under CUDA 12.2, ensuring high-efficiency parallelism for deep learning frameworks during the gradient updates of MAPPO or $\text{MA}^2\text{E}$.

\end{document}